\let\accentvec\vec
\documentclass[preprint,review,authoryear]{elsarticle}

\let\vec\accentvec
\usepackage[utf8]{inputenc}
\usepackage[T1]{fontenc}
\usepackage{textcomp,bold-extra}
\usepackage{amsmath,amsfonts,amssymb,relsize,ifthen}
\usepackage{multirow,multicol,booktabs,wrapfig,array}
\usepackage[english]{babel}
\usepackage[svgnames]{xcolor}
\usepackage{microtype}

\usepackage{numprint}
\npthousandsep{\,}\npthousandthpartsep{}\npdecimalsign{.}

\usepackage{tikz}
\usetikzlibrary{arrows,automata}

\usepackage{url}

\usepackage{pgfplots}
\pgfplotsset{compat=1.3}
\usetikzlibrary{patterns}
\usetikzlibrary{arrows}

\usepackage{listings}

\usepackage{mathtools}
\newtagform{special}[\texttt]{\textsf{(}}{\textsf{)}}
\usetagform{special}

\usepackage[inline]{enumitem}

\usepackage{color}
\usepackage{stmaryrd}
\usepackage{xspace}

\newcommand{\sqqcupop}{%
  \begin{tikzpicture}
    \draw[line width=0.09ex] (0,0.32ex) -- (0.98ex+0.14ex,0.32ex);
    \pgfsetroundcap
    \draw[line width=0.09ex,rounded corners=0.003ex] (0,1.05ex+0.14ex) -- (0,0) -- (0.98ex+0.14ex,0) -- (0.98ex+0.14ex,1.05ex+0.14ex);
  \end{tikzpicture}}

\newcommand{\sqqcapop}{%
  \begin{tikzpicture}
    \draw[line width=0.09ex] (0,0.74ex+0.14ex) -- (0.98ex+0.14ex,0.74ex+0.14ex);
    \pgfsetroundcap
    \draw[line width=0.09ex,rounded corners=0.003ex] (0,0) -- (0,1.05ex+0.14ex) -- (0.98ex+0.14ex,1.05ex+0.14ex) -- (0.98ex+0.14ex,0);
  \end{tikzpicture}}

\newcommand{\updop}{%
  \begin{tikzpicture}
    \draw[line width=0.09ex] (0,0.32ex) -- (0.98ex+0.14ex,0.32ex);
    \draw[line width=0.09ex] (0,0.74ex+0.14ex) -- (0.98ex+0.14ex,0.74ex+0.14ex);
    \pgfsetroundcap
    \draw[line width=0.09ex,rounded corners=0.003ex] (0,0) -- (0,1.05ex+0.14ex) -- (0.98ex+0.14ex,1.05ex+0.14ex) -- (0.98ex+0.14ex,0) -- (0,0);
  \end{tikzpicture}}

\newcommand{\boxop}{%
  \begin{tikzpicture}
    \pgfsetroundcap
    \draw[line width=0.09ex,rounded corners=0.003ex] (0,0) -- (0,1.05ex+0.14ex) -- (0.98ex+0.14ex,1.05ex+0.14ex) -- (0.98ex+0.14ex,0) -- (0,0);
  \end{tikzpicture}}

\DeclareMathOperator*{\sqqcap}{\hspace{0.5pt}\sqqcapop\hspace{1.3pt}}
\DeclareMathOperator*{\sqqcup}{\hspace{0.5pt}\sqqcupop\hspace{1.3pt}}
\DeclareMathOperator*{\upd}{\hspace{0.5pt}\updop\hspace{1.3pt}}

\DeclareMathOperator{\widen} {\sqqcup}
\DeclareMathOperator{\narrow}{\sqqcap}

\renewcommand{\qed}{{$\blacksquare$}}
\makeatletter
\let\Box\@undefined     
\makeatother
\DeclareMathOperator{\Box} {\boxop}

\newcommand{\update}{\mathit{update}}

\newcommand{\D}{{\mathbb D}}

\let\oldmarginpar\marginpar
\renewcommand\marginpar[1]{\-\oldmarginpar[\raggedleft #1]{\raggedright #1}}

\newcommand{\angl}[1]{\langle #1 \rangle}

\newcommand\justification[1]{\quad[\mbox{#1}]}

\catcode`\@=11        
\def\ignoretrue{\global\@ignoretrue}
\catcode`\@=12        

\def\defrelations{
   \def\better{\etc{\geq}}         
   \def\worse {\etc{\leq}}         %
   \def\equals{\etc{=}}            %
   \def\equivals{\etc{\equiv}}     %
   \def\lesseq{\etc{\leq}}         %
   \def\greateq{\etc{\geq}}        %
   \def\subseteqwithoutwhy{\etc{\subseteq}}
   \def\superseteq{\etc{\supseteq}}    %
   \def\inwithoutwhy{\etc{\in}}    %
   \def\iff{\etc{\mbox{iff}}}      
   \def\betterwhy##1{\etcwhy{\geq}{##1}}
   \def\worsewhy##1{\etcwhy{\leq}{##1}}
   \def\equalswhy##1{\etcwhy{=}{##1}}
   \def\equivalswhy##1{\etcwhy{\equiv}{##1}}
   \def\lesseqwhy##1{\etcwhy{\leq}{##1}}
   \def\greateqwhy##1{\etcwhy{\geq}{##1}}
   \def\subseteqwhy##1{\etcwhy{\subseteq}{##1}}
   \def\supseteqwhy##1{\etcwhy{\supseteq}{##1}}
   \def\sqsubseteqwhy##1{\etcwhy{\sqsubseteq}{##1}}
   \def\sqsupseteqwhy##1{\etcwhy{\sqsupseteq}{##1}}
   \def\iffwhy##1{\etcwhy{\mbox{iff}}{##1}}
   \def\implieswhy##1{\etcwhy{\Rightarrow}{##1}}
   \def\inwhy##1{\etcwhy{\in}{##1}}
   \def\LHS{{\rm LHS}}              
   \def\RHS{{\rm RHS}}              
}

{  \def\termskip{1ex}             
   \def\commentskip{0.5ex}        
   \def\etc##1{\cr\noalign{\vspace{\termskip}}##1&}
   \def\etcwhy##1##2{\cr\noalign{\vspace{\termskip}}##1&\justification{##2}\cr\noalign{\vspace*{\commentskip}}&}
    \bgroup
    \defrelations
    \tabskip=1em
    $$\halign to \displaywidth \bgroup\strut
        $\displaystyle\tabskip0pt{##}$\tabskip1ex&$\displaystyle\tabskip0pt{##}$\tabskip=0pt plus \textwidth&\tabskip=0pt plus 0pt\llap{##}\cr&}%
{\cr\egroup$$\egroup\ignoretrue}

{  \def\tabsign{&}
   \def\signspace{\ \ }  
   \def\etc##1{\tabsign\signspace ##1\signspace\tabsign\def\tabsign{}}
   \def\etcwhy##1##2
       { \tabsign\signspace
         \begin{array}[t]{@{}c@{}}
           ##1 \\[-0.4ex] \uparrow \\ \makebox[0mm]{[##2]}
         \end{array}
         \signspace
         \tabsign
         \def\tabsign{}
       }
   \def\linecalcendsequence{}
   \def\qed{\gdef\linecalcendsequence{\hfill\proofendsign}}     

   \defrelations
   \[ \begin{array}[t]{@{}l@{}c@{}l@{}}}%
{     \end{array} \]\linecalcendsequence\ignoretrue}%

{\global\def\csysendthing {#1} %
 \begin{equation*}%
 \setlength{\extrarowheight}{0.3em}%
 \begin{array}{rlr}}%
{\end{array}\csysendthing%
 \end{equation*}\ignorespacesafterend}

\newcommand{\scon}[2][]{ #2 & \ifthenelse{\equal{#1}{}}{}{\qquad\forall{#1}}}

\newcommand{\escon}[3][]
 {#2\!%
  \begin{aligned}[t]%
      #3%
  \end{aligned} \ifthenelse{\equal{#1}{}}{}{\qquad\forall{#1}}}
  
\newtheorem{fact}{Fact}
\newtheorem{theorem}{Theorem}
\newtheorem{lemma}[theorem]{Lemma}
\newtheorem{example}[theorem]{Example}
\newproof{proof}{Proof}

\allowdisplaybreaks

\newcommand{\eval}{\mathsf{eval}\xspace}
\newcommand{\widenarr}{\boxminus}
\newcommand{\SRR}{\textbf{SRR}\xspace}
\newcommand{\SW}{\textbf{SW}\xspace}
\newcommand{\SLR}{\textbf{SLR}\xspace}

\newcommand{\solve}{\textsf{solve}\xspace}
\newcommand{\infl}{\textsf{infl}\xspace}
\newcommand{\key}{\textsf{key}\xspace}

\newcommand{\dep}{\mathsf{dep}}

\newcommand{\wpoint}{\textsf{wpoint}\xspace}

\makeatletter
\global\let\tikz@ensure@dollar@catcode=\relax
\makeatother

\begin{document}
\begin{frontmatter}




\title{Efficiently intertwining widening and narrowing\tnoteref{conf,support}}

\author[pescara]{Gianluca Amato}
\ead{gamato@unich.it}

\author[pescara]{Francesca Scozzari}
\ead{fscozzari@unich.it}

\author[munchen]{Helmut Seidl}
\ead{seidl@in.tum.de}

\author[tartu]{Kalmer Apinis}
\ead{kalmer.apinis@ut.ee}

\author[tartu]{Vesal Vojdani}
\ead{vesal.vojdani@ut.ee}

\address[pescara]{Universit\`a di Chieti-Pescara}
\address[munchen]{Technische Universit\"at M\"unchen}
\address[tartu]{University of Tartu}

\begin{abstract}
  Non-trivial analysis problems require posets
with infinite ascending and descending chains. In order
to compute reasonably precise post-fixpoints of the resulting
systems of equations, Cousot and Cousot have suggested accelerated
fixpoint iteration by means of widening and narrowing 
\citep{CousotCousot76,CousotCousot77}.

The strict separation into phases, however, may unnecessarily give up
precision that cannot be recovered later,
as over-approximated interim results have to be fully propagated through the equation the system.
Additionally, classical two-phased approach is not suitable for equation systems with infinitely many unknowns---where demand driven solving must be used.
Construction of an intertwined approach must be able to answer when it is safe to apply narrowing---or when widening must be applied.
In general, this is a difficult problem.
In case the right-hand sides of equations are monotonic, however, we can always apply narrowing whenever we have reached a post-fixpoint for an equation.
The assumption of monotonicity, though, is not met in presence of widening.
It is also not met by equation systems corresponding
to context-sensitive inter-procedural analysis,
possibly combining context-sensitive analysis of local information
with flow-insensitive analysis of globals \citep{Apinis12}.

As a remedy, we present
a novel operator $\upd$ that combines a given widening operator $\widen$ with
a given narrowing operator $\narrow$.
We present adapted versions of round-robin as well as of 
worklist iteration,
local and side-effecting solving algorithms for the combined operator $\upd$ 
and prove that the resulting solvers always return sound results and
are guaranteed to terminate for monotonic systems whenever only
finitely many unknowns (constraint variables) are encountered.
Practical remedies are proposed for termination in the non-monotonic case.
Beyond that, we also discuss extensions of the base local solver that allow
to further enhance precision such as localized application of the operator $\upd$
and restarting of the iteration for subsets of unknowns.

\end{abstract}


\begin{keyword}
Static Program Analysis \sep Fixpoint Iteration \sep Constraint Solving \sep Widening and Narrowing \sep Termination
\end{keyword}


\tnotetext[conf]{This article extends and generalizes results presented by 
\citet{Apinis13} by integrating key ideas from \citet{amato2013localizing}.} 

\tnotetext[support]{This work was partially supported by the ARTEMIS Joint 
Undertaking under grant agreement n° 269335 and from the German Science 
Foundation (DFG). }


\end{frontmatter}

\section{Introduction}

From an algorithmic point of view, static analysis typically boils down to solving 
systems of equations over a suitable domain of values. The unknowns of the system
correspond to the invariants to be computed, e.g., for each program point or for
each program point in a given calling context or instance of a class.
For abstract interpretation, often 
complete lattices are chosen as domains of 
(abstract) values \citep{CousotCousot77}. 
Practically, though, partial orders can be applied which are not necessarily complete lattices---given only 
that they support an effective binary upper bound operation. 
This is the case, e.g., for polyhedra \citep{CousotHalbwachs78} or zonotopes 
\citep{DBLP:conf/cav/GhorbalGP09}.
Still,
variants of Kleene iteration can be applied to determine solutions. 
Right from the beginning of abstract interpretation, it also has been observed
that many interesting invariants are best expressed by domains
that have \emph{infinite} strictly ascending chains. 
Possibly infinite strictly ascending chains, though,
imply that naive Kleene iteration may not terminate. 
For that reason, Cousot and Cousot 
proposed a \emph{widening} iteration to obtain a valid invariant or, technically speaking, a \emph{post} solution
which subsequently may be improved by means of a \emph{narrowing} iteration
\citep{CousotCousot76,CC92}.
The widening phase can be considered as a Kleene iteration that is accelerated by means
of a widening operator to enforce that only finitely many increases of values occur for every
unknown. While enforcing termination, it may result in a crude over-approximation of the invariants
of the program.
In order to compensate for that, the subsequent narrowing iteration tries to improve a given
post solution by means of a downward fixpoint iteration, which again may be accelerated, in this case by means
of a 
 \emph{narrowing} operator.

Trying to recover precision once it has been thrown away, though, in general
is not possible (see, e.g., \citep{Halbwachs12} for a recent discussion).
Some attempts try to improve precision
by reducing the number of points where widening is
applied \citep{Cousot81,Bourdoncle93}, while others rely on refined 
widening or narrowing operators (see, e.g., \citep{Simon06,Cortesi11}). 
Recently, several authors have focused on methods to guide or stratify the exploration of the state space
\citep{gopan_guided_2007,gopan_lookahead_2006,gulavani_automatically_2008,stratified,henry_succinct_2012},
 
including techniques for automatic transformation of irregular loops 
\citep{gulwani_control-flow_2009,sharma_simplifying_2011}
or by restricting the use of widening to relevant parts of the program state 
only \citep{Halbwachs12}.

Our approach here at least partly encompasses those in 
\citep{Cousot81,Bourdoncle93} while it is complementary
to the other techniques and can, possibly, be combined with these.
Our idea is to avoid to postpone narrowing to a
second phase where all losses of information have already occurred and been
propagated. Instead, an attempt is made to systematically improve the current
information immediately by downward iterations.
This means that increasing and decreasing iterations are
applied in an \emph{interleaved} manner.
A similar idea has already applied in syntax-directed fixpoint iteration engines
as, e.g., 
in
the static analyzers {\sc Astr\'ee} 
\citep{blanchet2003static,CousotEtAl06-ASIAN}
and {\sc Jandom} \citep{amato2013localizing}.
In order to enforce termination, ad-hoc techniques such as restrictions to the
number of updates are applied.
Here, we explore such iteration strategies in a generic setting and provide sufficient
conditions when particular fixpoint algorithms are guaranteed to terminate.

The original formulation of narrowing as considered in 
\citep{CousotCousot76,CC92},
requires right-hand sides of equations to be monotonic so that the improving second
iteration phase is guaranteed to be downward.
Accordingly, the narrowing operator is only guaranteed to return meaningful results when applied 
in \emph{decreasing} sequences of values.
%
As we concentrate to the algorithmic side, we refer to these original notions of narrowing,
opposed to the more elaborate notions of \citep{CousotCousot92} which 
additionally take the 
\emph{concrete} semantics of the system to be analyzed into account.
Still, the assumption of monotonicity of right-hand sides, even dis-regarding 
the occurrences of widening and narrowing operators, may not always be met.
So,
monotonicity can no longer be guaranteed, 
when compiling context-sensitive  inter-procedural analysis 
into systems of equations \citep{FS99,Apinis12}.
Moreover,
the resulting equation systems may be \emph{infinite}
and thus can be handled by \emph{local} solvers only.
Local solvers query the value of an interesting unknown
and explore the space of unknowns only as much as required for answering the query. 
For this type of algorithm, the set of unknowns to be evaluated is not known beforehand. 
Accordingly, the values of unknowns may be
queried in the narrowing phase that have not yet been encountered before. As a
consequence, the rigid two-phase iteration strategy of one widening iteration followed
by one narrowing iteration can no longer be maintained. 

%



In order to cope with these obstacles, 
we introduce an operator $\upd$ which is a generic combination of a given widening 
$\widen$ with a given narrowing operator $\narrow$ and show that this new operator 
can be plugged into any generic solver of equation systems, 
be they monotonic or non-monotonic. The $\upd$ operator behaves like narrowing 
as long as the iteration is descending, and like widening otherwise. 
As a result, solvers are obtained that return reasonably precise post solutions
in one go---given that they terminate. 

Termination, though, is indeed an issue. We present two simple example systems
of monotonic equations where standard fixpoint algorithms such as round robin or
work-list iteration, when enhanced with the new operator, fail to terminate.
Therefore, we develop a variant of round robin as well as a variant of work-list iteration which in absence
of widening and narrowing are not or at least not much worse than their standard 
counter parts---but which additionally are guaranteed to terminate when the $\upd$-operator is applied 
to monotonic systems.

The idea of plugging the new operator $\upd$ into a generic \emph{local} solver works as well.
A local solver such as \citep{HKS10}, however, 
is not generic in the sense of the present paper---meaning that a naive enhancement 
with the operator $\upd$ is no longer guaranteed to return sound results.
As our main contribution, we therefore present
a variation of this algorithm which 
always returns a (partial) post solution and, moreover,
is guaranteed to terminate---at least for monotonic equation systems
and if only finitely many unknowns are encountered.
This algorithm relies on \emph{self-observation} not only for identifying
dependencies between unknowns on the fly, but also to determine a suitable
prioritization of the unknowns.
This vanilla version of a local iterator then is extended to cope with the
losses in precision detected in \citep{amato2013localizing}.
We present novel techniques for
\emph{localizing} the use of the operator $\upd$ to loop heads only.
These loop heads are dynamically detected and recomputed depending on the status of 
the fixpoint computation. Interestingly, dynamically recomputing loop heads during
fixpoint computation increases precision significantly.
As another improvement, we also considered to dynamically restart the
iteration for subsets of unknowns during a narrowing sub-iteration.
These algorithms then are extended to solvers
for \emph{side-effecting} constraint systems. Side-effecting systems allow to conveniently
specify analyses that combine context-sensitive analysis of local information
with flow-insensitive analysis of globals \citep{Apinis12} as provided, e.g.,
by the program analyzer {\sc Goblint} \citep{VV07}.
Since the different contributions to a global unknown are generated during the evaluation
of a subset of right-hand sides, which is not known before-hand and may vary
during fixpoint iteration, further
non-trivial changes are required to handle this situation.

The obstacle remains that termination guarantees in presence of 
unrestricted
non-monotonicity cannot be given. We attack this obstacle by two means. 
By practical experiments, we provide evidence that our iterator as is, not 
only terminates 
but is reasonably efficient --- at least 
for the equation systems of an inter-procedural interval analysis of several 
non-trivial real-world programs.  Secondly, we remark that, globally and 
irrespective of accidental experiments, termination can be guaranteed by 
\emph{bounding} for each unknown the number of switches from narrowing back 
to widening, or, more smoothly, to apply more and more aggressive narrowing 
operators. Note that this family of
restrictions is more liberal than restricting the number of updates of each 
unknown directly.

%

The rest of the paper is organized as follows.
In Section~\ref{s:chaos}, we present the concept of generic solvers.
In Section~\ref{s:narrowing}, we show that any such solver, when instantiated with $\upd$, 
returns a post solution of an arbitrary equation system (be it monotonic or not)
whenever the solver terminates. 
In order to enforce termination at least for 
finite systems of monotonic equations, we provide
in Section~\ref{s:term} new generic variants of round-robin iteration as well
as of work-list based fixpoint computation.
Section~\ref{s:local} introduces the new generic local $\upd$-solver \SLR,
which is subsequently enhanced with localization of $\upd$ (Section~\ref{s:localized})
and restarting (Section~\ref{s:restart}).
All three local solvers then are generalized to equation systems with side effects in Section~\ref{s:side}.
In section~\ref{s:experimental}, we compare the local solvers w.r.t.\ to precision and efficiency
within the analyzer framework {\sc Goblint}
and conclude in Section~\ref{s:concl}.

Sections~\ref{s:chaos} to~\ref{s:local} and the first part of section 
\ref{s:side}  are based on \citep{Apinis13}. The 
extension of ordinary and generic local solving provided in Sections~\ref{s:localized} and~\ref{s:restart}, 
as well as the second half of Section~\ref{s:side} are new. Also the experimental evaluation in 
Section~\ref{s:experimental} has been redone completely.

%


\section{Chaotic fixpoint iteration}\label{s:chaos}

Consider a system $S$ of equations $x = f_x$, for a set of unknowns $x\in X$, 
and over a set $\mathbb D$ of values
where the right-hand sides $f_x$ are mappings $(X\to \mathbb D)\to \mathbb D$.
Furthermore, let $\Box: \mathbb D\to \mathbb D\to \mathbb D$ be a binary 
operator to combine old values with the new contributions
of the right-hand sides. 

A $\Box$-solution of $S$ is an assignment $\rho:{\it X}\to \mathbb D$ such that 
for all unknowns $x\in {\it X}$,
\[ \rho[x] = \rho[x]\Box f_x\,\rho .\]
In the case that $\Box$ is defined as ${a\Box b} = b$, a $\Box$-solution is an ordinary solution of the system, i.e.,
a mapping $\rho$ with $\rho[x] = f_x\,\rho$ for all unknowns $x$.

Most of the time $\mathbb D$ is a \emph{directed set}, i.e., a poset such that 
for 
each pair of elements  $a, b \in \mathbb D$, there exists an upper bound $z$ 
such that 
$z \sqsupseteq a$ and $z \sqsupseteq b$. We denote by $a \sqcup b$ a generic 
upper bound of $a$ and $b$. In case $\mathbb D$ is a directed set, and the 
$\Box$-operator is an upper bound, a $\Box$-solution is a \emph{post} solution 
of the system, i.e., 
a mapping $\rho$ with $\rho[x] \sqsupseteq f_x\,\rho$ for all unknowns $x$. 
Likewise in case $\mathbb D$ is a \emph{downward}-directed set and $\Box$ is a 
lower 
bound,  a $\Box$-solution is a \emph{pre} solution of the system, 
i.e., a mapping $\rho$ with $\rho[x] \sqsubseteq f_x\,\rho$ for all unknowns 
$x$.

The operator $\Box$ can also be instantiated with widening and narrowing 
operators. According to \citep{CousotCousot76,CousotCousot77,CC92}, 
a widening operator $\widen$ for a poset $\mathbb D$ must satisfy that
$a \sqsubseteq a\widen b$, $b \sqsubseteq a\widen b$ for all $a,b\in 
\mathbb D$, and any widening sequence cannot be strictly ascending.
This implies that a $\widen$-solution then again provides a post 
solution of 
the original system $S$.
The situation is slightly more complicated for narrowing operators.
For a narrowing operator $\narrow$, $a\sqsupseteq b$ implies that ${a\sqsupseteq (a\narrow b) \sqsupseteq b}$ and any narrowing sequence cannot be strictly descending.
This means that narrowing can only be applied if the right-hand side of equations are
guaranteed to return values that are less than or equal to the values of the current left-hand sides.
Thus a mapping $\rho$ can only be a $\narrow$-solution if it is a post solution of 
the system.

\begin{figure}[t]
\begin{minipage}[b]{0.5\linewidth}
  \centering
  \begin{lstlisting}[xleftmargin=16mm]
    do 
      dirty := false;
      foreach x IN X do
        new := RHO[x] BOX f_x RHO;
        if RHO[x] NOTEQ new then
          RHO[x] := new;
          dirty := true;
    while (dirty)
  \end{lstlisting}
  \caption{The solver \textbf{RR}.}\label{f:rr}
\end{minipage}
\begin{minipage}[b]{0.5\linewidth}
  \centering
  \begin{lstlisting}[xleftmargin=18mm]
    W := X;
    while W NOTEQ EMPTYSET $\,$do 
      x := extract(W);
      new := RHO[x] BOX f_x RHO;
      if RHO[x] NOTEQ new then
        RHO[x] := new;
        W := W CUP  infl_x
    done
  \end{lstlisting}
  \caption{The Solver \textbf{W}.}\label{f:worklist}
\end{minipage}
\end{figure}

A (chaotic) solver for systems of equations is an algorithm that maintains a mapping
$\rho:{\it X}\to \mathbb D$ and performs a sequence of \emph{update steps},
starting from an initial mapping $\rho_0$.
Each update step selects an unknown $x$, evaluates the right-hand side $f_x$
of $x$ w.r.t.\ the current mapping $\rho_i$ and updates the value for $x$, i.e.,  
\[\rho_{i+1}[y]=
  \begin{cases}
    \rho_i[x]\Box f_x\,\rho_i, &\text{ if }x=y\\
    \rho_i[y],                 &\text{ otherwise. }
  \end{cases}\]
Then the algorithm is a $\Box$-solver if upon termination the final mapping (after completing $n$ steps) $\rho_n$ 
is a $\Box$-solution of $S$.
The algorithm is a generic solver if it works for any binary update operator 
$\Box$.
In this sense, the round-robin iteration of Fig.~\ref{f:rr} is a generic solver.
Note that, in most cases, we omit update step indices and, additionally, use imperative 
assignment syntax of the form ${\rho[x]\gets w}$ to change the value of the unknown $x$ to $w$ 
in the mapping $\rho$.

In order to prove that a given algorithm is a generic solver,
i.e., upon termination returns a $\Box$-solution, one typically verifies the
invariant that for every terminating run of the algorithm producing the sequence 
$\rho_0,\rho_1,\ldots,\rho_n$ of mappings, and every unknown $x$,
$\rho_i[x]\neq\rho_i[x]\Box f_x\,\rho_i$ implies that
for some $j\geq i$, an update $\rho_{j+1}[x]=\rho_j[x]\Box f_x\,\rho_j$
occurs.

Not every solver algorithm, though, may consider right-hand sides of equations
as black boxes, as the round-robin algorithm does.
The worklist algorithm from Fig.~\ref{f:worklist} can only be used as generic solver---given that
all dependences are provided before-hand. 
This means that for each right-hand side $f_x$ a (super-)set ${\sf dep}_x$ of unknowns
is given such that for all mappings $\rho,\rho'$,
$f_x\,\rho = f_x\,\rho'$ whenever $\rho$ and $\rho'$ agree on all unknowns in ${\sf dep}_x$. 
From these sets, we define the sets ${\sf infl}_y$ of unknowns possibly
influenced by (a change of the value of) unknown $y$, i.e.,
\[
{\sf infl}_y = \{x\in X\mid y\in{\sf dep}_x\}\cup\{y\}\,.
\]
In the case that the value of some unknown $y$ changes, all right-hand sides of 
unknowns in the set ${\sf infl}_y$ must be re-computed.
Note that whenever an update to an unknown $y$ provides a new value, we 
re-schedule $y$ for evaluation as well. This is a precaution for
the case that the operator $\Box$ is \emph{not} (right) idempotent.
Here, an operator $\Box$ is called \emph{idempotent} if the following
equality:
\[
(a\;\Box\;b)\;\Box\;b = a\;\Box\;b
\]
holds for all $a,b$.
In this sense, the operators $\sqcup$ and $\sqcap$ are idempotent and often also 
$\widen$ and $\narrow$. An operator such as $\frac{a+b}{2}$, however,
for $a,b\in{\mathbb R}$ is not idempotent.

\section{Enhancing Narrowing}\label{s:narrowing}

First, we observe:

\begin{fact}
Assume that all right-hand sides of the system $S$ of equations over a poset 
$\mathbb D$ are monotonic and that $\rho_0$ is a post solution of $S$, and 
$\Box$ is a narrowing operator $\sqqcap$.
Then the sequence $\rho_0,\rho_1,\ldots$ of mappings 
produced by a generic $\Box$-solver, is defined
and decreasing.
\qed
\end{fact}

\noindent
Thus, any generic solver can be applied to improve
a post solution by means of a narrowing iteration---given 
that all right-hand sides of equations are monotonic.

\medskip

Equation systems for context-sensitive inter-procedural analysis, though,
are not necessarily monotonic. In the following we show how to lift
the technical restrictions to the applicability of narrowing.
Given a widening operator $\widen$ and a narrowing operator $\narrow$,
we define a new binary operator $\upd$ by:
\begin{align*}
  a\upd b &= 
  \begin{cases}
    a \sqqcap b, &\text{ if } b\sqsubseteq a\\
    a \sqqcup b,  &\text{ otherwise. }
  \end{cases}
\end{align*}
Note that the operator $\upd$ is not necessarily idempotent, but
whenever narrowing is idempotent the following holds:
\begin{align*}
(a\;\upd\;b)\;\upd\;b &= (a\;\upd\;b)\;\narrow\;b\\
\intertext{and therefore also}
((a\;\upd\;b)\;\upd\;b)\;\upd\;b &= (a\;\upd\;b)\;\upd\;b\,.
\end{align*}
A fixpoint algorithm equipped with the operator $\upd$ applies 
widening as long as values grow. Once the evaluation of the
right-hand side of a unknown results in a smaller or equal value,
narrowing is applied and values may shrink.
For the operator $\upd$, we observe:

\begin{lemma}
Consider a finite system $S$ of equations over a directed set $\mathbb D$.
Then every $\upd$-solution $\rho$ of $S$ is a post solution, i.e., 
for all unknowns $x$, $\rho[x] \sqsupseteq f_x\,\rho$.
\end{lemma}

\begin{proof}
Consider a mapping $\rho$ that is a $\upd$-solution
of $S$ and an arbitrary unknown $x$. 
For a contradiction assume that $\rho[x] \not\sqsupseteq f_x\,\rho$.
But then we have:
\[
\rho[x] \;=\; \rho[x]\upd f_x\,\rho \;=\; \rho[x]\widen f_x\,\rho \;\sqsupseteq\; f_x\,\rho
\]
in contradiction to our assumption!
Accordingly, $\rho$ must be a post solution of the system of equations $S$.
\qed
\end{proof}

\noindent
Thus, every generic solver for directed sets $\mathbb D$ can be turned into a 
solver computing post solutions by using the
combined widening and narrowing operator.
%
The intertwined application of widening and narrowing, which naturally
occurs when solving the system of equations by means of $\upd$, 
has the additional advantage that values may also \emph{shrink} in-between.
Improving possibly too large values, thus, may take place immediately
resulting in overall smaller, i.e., better post solutions.
Moreover, no restriction is imposed any longer concerning monotonicity
of right-hand sides.

\section{Enforcing termination}\label{s:term}

For the new operator $\upd$, termination 
cannot generally be guaranteed for all solvers. In this section, we therefore present a
modification of worklist iteration which is guaranteed to terminate---given 
that all right-hand sides of equations are monotonic.

\noindent
\begin{example}\label{e:term}
Consider the system:
\[
\begin{array}{lll}
x_1 &=& x_2 \\
x_2 &=& x_3 + 1 \\
x_3 &=& x_1
\end{array}
\]
with $\mathbb D={\mathbb N}\cup\{\infty\}$, the lattice of non-negative 
integers, equipped with the
natural ordering $\sqsubseteq$ given by $\leq$ and extended with $\infty$.
Consider a widening $\widen$ where
${a\widen b = a}$ if ${a= b}$ and ${a\widen b = \infty}$ otherwise,
together with a narrowing $\narrow$ where, for ${a\geq b}$,
 ${a\narrow b = b}$ if ${a=\infty}$, and ${a\narrow b = a}$ otherwise.
%
Round-robin iteration with the operator $\upd$ for this system starting from the mapping
${\rho_0 = \{x_1\mapsto 0, x_2\mapsto 0, x_3\mapsto 0\}}$,
will produce the following sequence of mappings:
	\[
        \setlength{\arraycolsep}{5pt}
	\begin{array}{c|c|c|c|c|c|c|c}
            & 0 & 1      & 2      & 3      & 4      & 5      &       \\
	\hline
        x_1 & 0 & 0      & \infty & 1      & \infty & 2      & \dots \\
        x_2 & 0 & \infty & 1      & \infty & 2      & \infty & \dots \\
        x_3 & 0 & 0      & \infty & 1      & \infty & 2      & \dots \\
	\end{array}
	\]
Iteration does not terminate---although right-hand sides are monotonic. \qed
\end{example}

\noindent
A similar example shows that ordinary worklist iteration,
enhanced with $\upd$, also may not terminate, even if all equations
are monotonic.

\begin{example}\label{e:w}
Consider the two equations:
\begin{align*}
x_1 &= (x_1+1) \sqcap (x_2+1)	\\
x_2 &= (x_2+1) \sqcap (x_1+1)	
\end{align*}
using the same lattice as in Example~\ref{e:term}
where $\sqcap$ denotes minimum, i.e., the greatest lower bound.
Assume that 
the work-set is maintained with a lifo discipline. For $W = [x_1,x_2]$, worklist iteration,
starting with the initial mapping $\rho_0=\{x_1\mapsto 0, x_2\mapsto 0\}$,
results in the following iteration sequence:
\[
\begin{array}{l|c|c|c|c|c|c|c|c|c}
  W\!\!
&\scalebox{0.7}{\!\!$[x_1,x_2]$\!\!}
&\scalebox{0.7}{\!\!$[x_1,x_2]$\!\!}
&\scalebox{0.7}{\!\!$[x_1,x_2]$\!\!}
&\scalebox{0.7}{$[x_2]$}
&\scalebox{0.7}{\!\!$[x_2,x_1]$\!\!}
&\scalebox{0.7}{\!\!$[x_2,x_1]$\!\!}
&\scalebox{0.7}{$[x_1]$}
&\scalebox{0.7}{\!\!$[x_1,x_2]$\!\!} \\ \hline
x_1 & 0 & \infty & 1 & 1 & 1   & 1 & 1 & \infty & \ldots\!\\
x_2 & 0 & 0   & 0 & 0 & \infty & 2 & 2 & 2   & \ldots\!	
\end{array}
\]
which does not terminate.
\qed
\end{example}

We present modified versions of the round-robin solver as
well as the worklist solver for which termination can be guaranteed.
The worst case complexity for the new round-robin solver turns out to be faster than ordinary round-robin iteration, even
by a factor of 2.
For the new worklist solver, theoretical complexity is
at least not far away from the classical iterator.

For both algorithms, we assume that we are given a fixed \emph{linear ordering}
on the set of unknowns so that $X = \{x_1,\ldots,x_n\}$. 
The ordering will affect the iteration strategy, and therefore,
as shown by \citet{Bourdoncle90}, has a significant impact on performance. 
Hence, the linear ordering should be chosen in a way that innermost loops would be 
evaluated before iteration on outer loops. 
For unknowns $x_i$ and the system of equations given by
$x_i = f_i$, for $i=1,\ldots,n$, the new round-robin algorithm
is shown in Fig.~\ref{f:srr}.

\begin{wrapfigure}{r}{0.35\textwidth}
  \centering
  \begin{lstlisting}[xleftmargin=6mm]
    void solve(i) {
      if i = 0 then return;
      solve(i-1);
      new := RHO[x_i] BOX f_i RHO;
      if RHO[x_i] NOTEQ new then
        RHO[x_i] := new;
        solve(i);
    }
  \end{lstlisting}
  \caption{The new solver \textbf{SRR}.}\label{f:srr}
\end{wrapfigure}

Let us call the new algorithm {\bf SRR} (\emph{structured} round-robin).
For a given initial mapping $\rho_0$, structured round-robin
is started by calling ${\sf solve}\,n$.
The idea of the algorithm is, when called for a number $i$,
to iterate on the unknown $x_i$ until stabilization.
Before every update of the unknown $x_i$, however, 
all unknowns $x_j,j<i$ are recursively solved.
Clearly, the resulting algorithm is a generic $\Box$-solver.

Recall that a poset ${\mathbb D}$ has height $h$ if $h$
is the maximal length of a strictly increasing chain
$d_0 \sqsubset d_1\sqsubset\ldots\sqsubset d_h$.
We find:

\begin{theorem}
\label{th:SRR}
Consider a system of $n$ equations
over a directed set $\mathbb D$ where all right-hand sides are monotonic.
Then the following holds for the algorithm {\bf SRR} when started on any mapping:
\begin{enumerate}
\item	If $\mathbb D$ has bounded height $h$ and ${\Box = 
\sqcup}$, then
	{\bf SRR} terminates with a post-solution
	after at most $n+\frac{h}{2}n(n+1)$ evaluations
	of right-hand sides $f_i$.
\item In presence of possibly unbounded ascending chains, when instantiated   
with $\Box=\upd$, {\bf SRR} terminates with a post-solution.
\end{enumerate}
\end{theorem}

\noindent
The first statement indicates that {\bf SRR} may favorably compete with
ordinary round robin iteration in case that no widening and narrowing is required.
The second statement, on the other hand, provides us with a termination guarantee --- 
whenever only all right-hand sides are monotonic.

\begin{proof}
Recall that ordinary round robin iteration for directed sets of bounded height performs
at most $h\cdot n$ rounds due to increases of values of unknowns
plus one extra round to detect termination, giving in total
\[
n+ h\cdot n^2
\]
evaluations of right-hand sides. 
In contrast for structured round robin iteration, termination for
unknown $x_i$ requires one evaluation when ${\sf solve}\:i$ is called for the
first time and then one further evaluation for every update of one of the
unknowns $x_n,\ldots, x_{i+1}$. This sums up to ${h\cdot(n-i)+1}$ evaluations 
throughout the whole iteration.
This gives a overhead of
\[
n + h\cdot\sum_{i=1}^n (n-i) =
n + \frac{h}{2}\cdot n\cdot(n-1) \ .
\]
Additionally, there are $h \cdot n$ evaluations that
increase values. In total, the number of evaluations,
therefore, is
\[
n + \frac{h}{2}\cdot n\cdot(n-1) + h\cdot n =
n+ \frac{h}{2} \cdot n\cdot(n+1)
\]
giving us statement 1.

For the second statement, we proceed by induction on $i$. 
The case $i=0$ is vacuously true.  For the inductive step, assume that $i>0$ and 
for all $j<i$, ${\sf solve}\;j$ terminates for any mapping.
To arrive at a contradiction, assume that ${\sf solve}\;i$ for the current
mapping $\rho$ does not terminate.
First, consider the case where $f_i\;\rho$ returns a value smaller than
$\rho[x_i]$. Since \textbf{SRR} is a generic solver, we have for all $j<i$, 
${\rho[x_j] = \rho[x_j]\upd f_j\;\rho}$, implying that
$\rho[x_j]\sqsupseteq f_j\;\rho$ for all $j<i$. After $\rho[x_i]$ is updated, 
by monotonicity, it still holds that $\rho[x_j]\sqsupseteq f_j\;\rho$ for
all $j<i$. Solving for the unknown $i-1$ will only cause further 
descending steps, where $\upd$ behaves like $\narrow$. The subsequent 
iteration of ${\sf solve}\;i$ will produce a decreasing sequence of mappings.
Since all decreasing chains produced by narrowing are ultimately stable,
the call ${\sf solve}\;i$ will terminate---in contradiction to our assumption.

Therefore, non-termination is only possible if during the whole run of ${\sf solve}\;i$, 
evaluating $f_i\;\rho$
must always return a value that is not subsumed by $\rho[x_i]$.
Since all calls ${\sf solve}\;(i-1)$ in-between terminate by the induction hypothesis, 
a strictly increasing sequence of values for $x_i$ is obtained 
that is produced by repeatedly applying the widening operator. 
Due to the properties of widening operators, any such sequence is eventually
stable---again in contradiction to our assumption.
We thus conclude that ${\sf solve}\;i$ is eventually terminating.
\qed
\end{proof}

\noindent
\begin{example}
Recall the equation system, for which round-robin iteration did not terminate. 
With structured round-robin iteration, however, we 
%
obtain the following sequence of updates:
\vspace{0.2em}
\[
\begin{array}{c|c|c|c|c|c|c|c|c|}
i            &\;\;\;\;& 2      & 1      & 2      & \,1\,\, & 3      & 2      & 1      \\ \hline
x_1          & 0& 0      & \infty & \infty & 1 & 1      & 1      & \infty \\
x_2          & 0& \infty & \infty & 1      & 1 & 1      & \infty & \infty \\ 
x_3          & 0& 0      &      0 & 0      & 0 & \infty & \infty & \infty 
\end{array}
\vspace{0.3em}
\]
where the evaluations of unknowns not resulting in an update
have been omitted.
Thus, structured fix-point solving quickly stabilizes for this example. \qed

\end{example}

The idea of structured iteration can also be lifted to worklist
iteration.
Consider again a system ${x_i = f_i}$, for ${i=1,\ldots,n}$, of equations.
As for the ordinary worklist algorithm, we assume that  
for each right-hand side $f_i$ a (super-)set ${\sf dep}_i$ of unknowns
is given, such that for all mappings $\rho,\rho'$,
$f_i\,{\rho = f_i\,\rho'}$ whenever $\rho$ and $\rho'$ agree on all unknowns 
in ${\sf dep}_i$. As before for each unknown $x_j$, let ${\sf infl}_j$ denote 
the set consisting of the unknown $x_j$ together with all unknowns influenced by
$x_j$.
Instead of a plain worklist, the modified algorithm maintains the set of 
unknowns to be reevaluated, within a \emph{priority queue} $Q$.
In every round, not an arbitrary element is extracted from $Q$ --- but the unknown
with the least index.
The resulting algorithm is presented in Fig.~\ref{f:sw}.

\begin{wrapfigure}{r}{0.42\textwidth}
  \centering
  \begin{lstlisting}[xleftmargin=7mm]
    Q := EMPTYSET;
    for i = 1 to n do add Q x_i;
    while Q NOTEQ EMPTYSET do
      x_i := extract_min(Q);
      new := RHO[x_i] BOX f_i RHO;
      if RHO[x_i] NOTEQ new then
        RHO[x_i] := new;
        foreach x_j IN infl_i $\;$do 
          add Q x_j
    done
  \end{lstlisting}
  \caption{The new solver \textbf{SW}.}\label{f:sw}
\end{wrapfigure}

Here, the function ${\sf add}$ inserts an element into the priority queue
or leaves the queue unchanged if the element is already present.
Moreover, the function ${\sf extract\_min}$ removes the unknown
with the smallest index from the queue and returns it as result.

Let us call the resulting algorithm {\bf SW} (structured worklist iteration).
Clearly, the resulting algorithm is a generic solver for systems of equations where
the dependences between unknowns are explicitly given.

\noindent
\begin{example}\label{e:w-1}
Consider again the system from example \ref{e:w}.
Structured worklist iteration using $\upd$
for this system results in the following iteration:
\[
\begin{array}{l|c|c|c|c|c|c|c|c|}
Q
&\scalebox{0.7}{\!\!$[x_1,x_2]$\!\!}
&\scalebox{0.7}{\!\!$[x_1,x_2]$\!\!}
&\scalebox{0.7}{\!\!$[x_1,x_2]$\!\!}
&\scalebox{0.7}{\;$[x_2]$\;}
&\scalebox{0.7}{\!\!$[x_1,x_2]$\!\!}
&\scalebox{0.7}{\!\!$[x_1,x_2]$\!\!}
&\scalebox{0.7}{\;$[x_2]$\;}
&\scalebox{0.7}{$[]$}	\\ \hline
x_1 & 0 & \infty & 1 & 1 & 1   & \infty  & \infty & \infty \\
x_2 & 0 & 0   & 0 & 0 & \infty & \infty  & \infty & \infty
\end{array}
\]
and thus terminates.
\qed
\end{example}

\noindent
In general, we have:

\begin{theorem}
\label{th:SW}
Assume the algorithm {\bf SW} is applied to a system of equations
over a directed set ${\mathbb D}$ and
that each right-hand side is monotonic.
\begin{enumerate}
\item	Assume that the maximal length of a strictly
        ascending chain is bounded by $h$.
	When instantiated with ${\Box=\sqcup}$,  
	the algorithm terminates after at most
	${h\cdot{}N}$ evaluations of right-hand sides where
	$N{=}\sum_{i=1}^n (2+ \left|{\sf dep}_i\right|)$.
\item	When instantiated with $\Box = \upd$ 
	and started on any mapping, the algorithm 
	is guaranteed to terminate and, thus, always to return
	a post solution.
\end{enumerate}
\end{theorem}

\noindent
The first statement of the theorem indicates that {\bf SW}
behaves complexity-wise like ordinary worklist iteration:
in case that the directed set ${\mathbb D}$ has finite height,
the only overhead to be paid for is an extra logarithmic factor for maintaining
the priority queue.
The second statement, perhaps, is more surprising: it provides
us with a termination guarantee for arbitrary lattices and the operator $\upd$ --- 
whenever only all right-hand sides are monotonic.

\begin{proof}
We proceed by induction on the number $n$ of unknowns. The case $n=1$ is true
by definition of widening and narrowing.
For the induction step assume that the assertion holds for systems of equations
of $n-1$ unknowns. Now consider a system of equations for a set $X$ of
cardinality $n$, and assume that $x_n$ is the unknown which is
larger than all other unknowns in $X$.

For a contradiction assume that {\bf SW} does not terminate for the system
of equations for $X$.
First assume that the unknown $x_n$ is extracted from the queue $Q$ only finitely many
times, say $k$ times where $d$ is the last value computed for $x_n$.
This means that after the last extraction, 
an infinite iteration occurs on the subsystem on the unknowns 
$X' = X\setminus\{n\}$ where for $x_r\in X'$, the right-hand side
is given by $f'_r\,\rho = f_r\,(\rho\oplus\{x_n\mapsto d\})$.
By inductive hypothesis, however, the algorithm {\bf SW} for this
system terminates --- in contradiction to our assumption.

Therefore, we may assume that the unknown $x_n$ is extracted infinitely often from 
$Q$. Let $\rho_i,i\in{\mathbb N},$ denote the sequence of mappings
at these extractions. Since $Q$ is maintained as a priority queue, 
we know that
for all unknowns $x_r$ with $r < n$, the inequalities $\rho_i[x_r]\sqsupseteq f_r\,\rho_i$ hold.
Let $d_i = \rho_i[x_n]$. 
If for any $i$, $f_n\,\rho_i \sqsubseteq d_i$, the next value $d_{i+1}$ for $x_n$ then is
obtained by $d_{i+1} = d_i\narrow f_n\,\rho_i$ which is less or equal to $\rho_i$.
By monotonicity, this implies that in the subsequent iteration, the values for
all unknowns $x_r, r\leq n,$ may only decrease. 
The remaining iteration is a pure
narrowing iteration and therefore terminates.
In order to obtain an infinite sequence of updates for $z$, we conclude that
for no $i$, $f_n\,\rho_i \sqsubseteq d_i$. Hence for every $i$,
$d_{i+1} = d_i\widen f_n\,\rho_i$ where $d_i\sqsubseteq d_{i+1}$.
This, however, is impossible due to the properties of the widening operator.
In summary, we conclude that $x_n$ is extracted only finitely often from $Q$.
Hence the fixpoint iteration terminates.
\qed
\end{proof}
	
\noindent
The algorithm {\bf SW} can
also be applied to non-monotonic systems. There, however, termination
can no longer be guaranteed. 
In fact, the assumption of monotonicity is not a defect of our solvers {\bf SRR} or {\bf SW}, 
but inherent to \emph{any} terminating fixpoint iteration which intertwines widening and narrowing.

\noindent
\begin{example}\label{e:non}
Consider the single equation:
\[
	{\sf x} = {\bf if}\;({\sf x} = 0)\;{\bf then}\;1\;{\bf else}\;0
\]
over the lattice of naturals (with infinity) with $a\widen b = \infty$ whenever $a < b$ and
$a\narrow b = b$ whenever $a=\infty$. The right-hand side of this equation is 
not monotonic. 
An iteration, defined by ${\sf x}_0 = 0$ and ${\sf x}_{i+1} = {\sf x}_i \upd f {\sf x}_i$ for $i\geq 0$ ($f$ the right-hand
side function of the equation) will produce the sequence:
\[
        0 \to \infty \to 0 \to \infty \to 0 \to \infty\to\ldots
\]
and thus will not terminate. We conclude that in absence of monotonicity, we cannot hope for termination---at
least, without further assumptions on the right-hand sides of the equations.
\qed
\end{example}
Still, there is one generic idea to enforce termination for \emph{all} $\upd$-solvers and all monotonic or non-monotonic
systems of equations.
This idea is to equip each unknown with a separate counter that counts how often the solver
has switched from narrowing back to widening. This number then may be taken
into account by the $\upd$-operator, e.g., by choosing successively less aggressive
narrowing operators $\narrow_{\!0},\narrow_{\!1},\ldots $, and, ultimately, to give up improving the
obtained values. 
The latter is achieved by defining 
$a\narrow_{\!k} b = a$ for a certain threshold $k$.

\section{Local generic solvers}\label{s:local}

Similar to generic solvers, we define generic \emph{local} solvers. 
Use of local solvers can be considered if systems of equations are infeasibly large
or even infinite. Such systems are, e.g., encountered for context-sensitive
analysis of procedural languages \citep{CousotCousot77c,Apinis12}.
Local solvers query the system of equations for the value of
a given unknown of interest and try to evaluate only the right-hand sides
of those unknowns that are needed for answering the query
\citep{LeCharlierHent92,VergauwenWaumanLewi94,FS99}. 
For that, it seems convenient that the \emph{dynamic} dependences between unknowns
are approximated.
For a mapping $\rho$, a set ${X'\subseteq{\it X}}$
subsumes all dynamic dependences of a function ${f:({\it X}\to \D)\to \D}$ 
(w.r.t.\ $\rho$) in the case that $f\,\rho = f\,\rho'$ whenever $\rho'|_{X'} = \rho|_{X'}$.
Such sets can be constructed on the fly whenever the
function $f$ is \emph{pure} in the sense of \citep{DBLP:conf/icalp/HofmannKS10}.

Essentially, purity for a right-hand side $f$ means that
evaluating $f\,\rho$ for a mapping $\rho$ operationally consists
of a finite sequence of value lookups in $\rho$ where the next unknown whose value has to be
looked up may only depend on the values that have already been queried.
Once the sequence of lookups has been completed, the final value is determined
depending on the sequence of values and finally returned.

In this case, the set $X'$ can be chosen as the set of all variables $y$  for which the value $\rho\;y$ is queried when evaluating (an implementation of) the function $f$ for the argument $\rho$. Let us denote this set by $\dep_x\;\rho$.

A \emph{partial} $\Box$-solution of an (infinite) system of pure equations $S$ is a
set ${\sf dom}\subseteq{\it X}$ and a
mapping $\rho:{\sf dom}\to \D$  
with the following two properties:
\begin{enumerate}
\item	$\rho[x] = \rho[x]\Box f_x\,\rho$ for all $x\in{\sf dom}$; and
\item	${\sf dep}_x\;\rho \subseteq{\sf dom}$ for all $x\in{\sf dom}$
\end{enumerate}
In essence, this means that a partial $\Box$-solution is a $\Box$-solution of the
subsystem of $S$ restricted to unknowns in ${\sf dom}$.

\noindent
\begin{example}\label{e:loc}
  The following equation system (for ${n\in\mathbb{N}=\D}$)
  \begin{align*}
    y_{2n}   &= \max(y_{y_{2n}},n)\\
    y_{2n+1} &= y_{6n+4}
  \end{align*}
  is infinite as it uses infinitely many unknowns, but has at least one finite 
  partial $\max$-solution---the set $\mathsf{dom} = \{y_1,y_2,y_4\}$ together with the mapping 
  ${\rho=\{y_1\mapsto 2, y_2 \mapsto 2, y_4\mapsto 2\}}$ where 
  ${\sf dep}_{y_1}\;\rho = \{y_4\}$, ${\sf dep}_{y_2}\;\rho = \{y_2\}$ and
   ${\sf dep}_{y_4}\;\rho = \{y_4, y_2\}$.
  \qed
\end{example}


A \emph{local} generic solver instantiated with an operator $\Box$, then, is an algorithm
that, when given a system of pure equations $S$, an initial mapping $\rho_0$ 
for all unknowns, and an unknown $x_0$,
performs a sequence of update operations that, upon termination, results in a
partial $\Box$-solution $(\mathsf{dom},\rho)$, such that $x_0\in\mathsf{dom}$.
In practice, when it is possible, $\rho_0$ is chosen to map each unknown to 
the least element of the directed set.

At first sight, it may seem surprising that such local generic solvers
may exist. In fact, one such instance can be derived from the round-robin algorithm.
For that, the evaluation of right-hand sides is instrumented in such a way that it keeps
track of the set of accessed unknowns. Each round then operates on a growing set
of unknowns. In the first round, just $x_0$ alone is considered. In any subsequent round
all unknowns are added whose values have been newly accessed during the last iteration.

\begin{figure}
\begin{lstlisting}[xleftmargin=1cm,multicols=2]
    let rec solve x = 
      if x NOTIN stable then  
        stable := stable CUP {x};
        tmp := RHO[x] SQCUP f_x (eval x);
        if tmp NOTEQ RHO[x] then 
          W := infl[x];
          RHO[x] := tmp;
          infl[x] := EMPTYSET; 
          stable := stable \ W; 
          foreach x IN W do solve x
        end
      end
  
    and eval x y = 
      solve y;
      infl[y] := infl[y] CUP {x}; 
      RHO[y]

    in
      stable := EMPTYSET; 
      infl := EMPTYSET; 
      RHO := RHO$_0$; 
      solve x$_0$; 
      RHO
\end{lstlisting}  
\caption{The solver \textbf{RLD} from \citep{HKS10}.}\label{f:rld}
\end{figure}

A more elaborate algorithm for local solving is formalized
by \citet{HKS10}, namely the solver \textbf{RLD}, as shown in Figure~\ref{f:rld}. 
This algorithm has the benefit of visiting nodes in a more efficient order, 
first stabilizing innermost loops before iterating on outer loops. 
The global assignment ${\sf infl}: X \to 2^{\it X}$ 
records, for each encountered unknown $y$, the set of unknowns $x\in{\sf dom}$ 
with the following two properties:
\begin{itemize}
\item the last evaluation of $f_x$ has accessed the unknown $y$;
\item	since then, the value of the unknown $y$ has not changed.
\end{itemize}

The right-hand sides $f_x$ are not directly evaluated for the current mapping 
$\rho$,
but instead for a helper function ${\sf eval}$ which in the end, 
returns values for unknowns. Before that, however, the helper function ${\sf 
eval}$ provides extra book-keeping of the encountered dependence between 
unknowns. In order to be able to track dependences between unknowns, the 
helper function ${\sf eval}$ receives as a first argument the unknown $x$ 
whose right-hand side is under evaluation.
The function ${\sf eval}$ first computes the best possible value for 
$y$ by calling the procedure ${\sf solve}$ for $y$.
Then ${\sf eval}$ records the fact that $x$ depends on $y$,
by adding $x$ to the set ${\sf infl}[y]$. Only then is
the corresponding value $\rho[y]$ returned.

The main fixpoint iteration is implemented by the procedure ${\sf solve}$.
It requires a set ${\sf stable}$ of unknowns such that, if $x$ is in  ${\sf 
stable}$, a call to the procedure ${\sf solve}\,x$ has been started and no 
unknowns influencing $x$ have been updated.

This algorithm correctly determines a post-solution of the set of 
equations upon termination. However, when enhanced with an operator $\Box$, it 
is \emph{not} a generic solver in our sense, since it is not guaranteed to 
execute as a sequence of \emph{atomic} updates. 
Due to the recursive call to procedure ${\sf solve}$ at the beginning of ${\sf eval}$,
one evaluation of a right-hand side may occur 
nested into the evaluation of another right-hand side.
Therefore, conceptually, it may happen that an evaluation of a right-hand side  
uses the values of unknowns from several different mappings $\rho_i$ 
from the sequence $\rho_0,\rho_1,\ldots,\rho_n$, instead of the latest mapping $\rho_n$.
Accordingly, the solver \textbf{RLD} is not
guaranteed to return a $\Box$-solution---even if it terminates.
We therefore provide a variant of \textbf{RLD} where 
right-hand sides (conceptually) are executed atomically. 

Clearly, a local generic solver does not terminate if
infinitely many unknowns are encountered.
Therefore, a reasonable local solver will try to consider as few unknowns as possible. 
Our solver, thus, explores the values of unknowns by recursively descending
into solving unknowns \emph{newly} detected while evaluating a right-hand side.
Certain equation systems, though, introduce infinite chains of dependences for the unknowns
of interest.
Those systems then cannot be solved by any local solver.
%
Here, we show that the new solver is guaranteed to terminate for the operator $\upd$
at least for equation systems which are monotonic and either finite or infinite but
where only finitely many unknowns are encountered. 

Let us call the new solver, on Fig.~\ref{f:slr}, {\bf SLR$_1$} 
(\emph{structured local recursive} solver).
The new algorithm maintains an explicit set ${\sf dom}\subseteq{\it X}$
of unknowns that have already been encountered. Beyond \textbf{RLD}, it 
additionally maintains a counter ${\sf count}$ which counts the number of 
elements in ${\sf dom}$, and a mapping ${\sf key}:{\sf dom}\to{\mathbb Z}$ 
that equips each unknown with its priority. Unknowns whose equations may 
possibly be no longer valid will be scheduled for reevaluation. This means 
that they are inserted into a global priority queue $Q$. 


\begin{figure}
\begin{lstlisting}[multicols=2,xleftmargin=2em]
    let rec solve x = 
      if x NOTIN stable then
        stable := stable CUP {x}
        tmp := RHO[x] BOX f_x (eval x);
        if tmp NOTEQ RHO[x] then
          W := infl[x] CUP {x};
          foreach y IN W do add Q y;
          RHO[x] := tmp;
          infl[x] := EMPTYSET;
          stable := stable MINUS $\;$W;
          while (Q NOTEQ EMPTYSET) $\land$ 
                (min_key Q <=key[x]) do
            solve (extract_min Q);
        end
      end
  
    and init y = 
      dom := dom CUP {y};
      key[y] := -count; 
      count++;
      infl[y] := {y};   
      RHO[y] := RHO$_0$[y]
  
    and eval x y = 
      if y NOTIN dom then
        init y;  
        solve y;
      end;
      infl[y] := infl[y] CUP {x};
      RHO[y]                                              

    in 
      stable := EMPTYSET; infl := EMPTYSET; 
      RHO := EMPTYSET; dom := EMPTYSET; 
      Q := empty_queue(); 
      count := 0; init x$_0$; 
      solve x$_0$; 
      RHO
      
      
      
      
      $\phantom{.}$
\end{lstlisting}  
\caption{The new solver \textbf{SLR}$_1$.}\label{f:slr}
\end{figure}

As in the algorithm \textbf{RLD}, right-hand sides $f_x$ are evaluated for a
helper function ${\sf eval}$. The function ${\sf eval}$ first checks whether 
the unknown $y$ is already contained in the domain ${\sf dom}$ of $\rho$.
If this is not the case, $y$ is first initialized by calling the
procedure ${\sf init}$. Subsequently, the best possible value for $y$ is 
computed by calling the procedure ${\sf solve}$ for $y$.

Initialization of a fresh unknown $y$ means that
$y$ is inserted into ${\sf dom}$ where it receives a key
less than the keys of all other unknowns in ${\sf dom}$.
For that, the variable ${\sf count}$ is used.
Moreover, ${\sf infl}[y]$ and $\rho[y]$ are initialized with $\{y\}$ and $\rho_0[y]$, 
respectively. 
%
Thus, the given function ${\sf eval}$ differs from the corresponding
function in \textbf{RLD} in that ${\sf solve}$ is recursively called only for 
\emph{fresh} unknowns, and also in that every unknown $y$ always depends on 
itself.

The main fixpoint iteration is implemented by the procedure ${\sf solve}$.
When ${\sf solve}$ is called for an unknown $x$, we assume that there is 
currently
no unknown $x'\in{\sf dom}$ with ${\sf key}[x']<{\sf key}[x]$ that violates 
its equation, i.e.,
for which ${\rho[x']\neq\rho[x']\;\Box\;f_{x'}\,\rho}$ holds.
In the procedure ${\sf solve}$ for $x$, 
%
the call ${\sf min\_key}\:Q$ returns the minimal key of an element in $Q$,
and ${\sf extract\_min}\;Q$ returns the unknown in $Q$ with minimal key and
additionally removes it from $Q$. Besides the global priority queue $Q$, the 
procedure ${\sf solve}$ also  requires the set ${\sf stable}$ as for 
\textbf{RLD}. Due to the changes in ${\sf eval}$ and the fact that $x$ is 
always added to $W$ during the execution of {\sf solve}~$x$, at each call of 
the procedure {\sf solve}, if $x \in {\sf stable}$ then either
\begin{itemize}
\item a call to the procedure ${\sf solve}\,x$ has been started and the update 
of $\rho[x']$ has not yet occurred; or
\item	the equality $\rho[x]=\rho[x]\;\Box\;f_{x}\,\rho$ holds.
\end{itemize}
The new function ${\sf solve}$ essentially behaves like the corresponding 
function in \textbf{RLD} with the notable exception that
not necessarily all unknowns that have been found unstable after the update of
the value for $x$ in $\rho$, are recursively solved right-away.
Instead, all these unknowns are inserted into the global priority queue $Q$ and then
${\sf solve}$ is only called for those unknowns $x'$ in $Q$ 
whose keys are less or equal than ${\sf key}[x]$. 
Since $x_0$ has received the largest key, 
the initial call ${\sf solve}\,x_0$ will result, upon termination, in an empty priority queue $Q$.

\noindent
\begin{example}
  Consider again the infinite equation system from example \ref{e:loc}.
  The solver $\mathbf{SLR}_1$, when solving for $y_1$, will return 
  the partial $\max$-solution
  ${\{y_0\mapsto 0, y_1\mapsto 2, y_2 \mapsto 2, y_4\mapsto 2\}}$.
  \qed
\end{example}

The modifications of the algorithm \textbf{RLD} to obtain 
algorithm {\bf SLR}$_1$
allow us not only to prove that it 
is a generic local solver, but also a strong result
concerning termination. 
Our main theorem is:

\begin{theorem}\label{t:local}
\begin{enumerate}
\item
When applied to any system of pure equations and interesting unknown $x_0$, the algorithm {\bf SLR}$_1$
returns a partial $\Box$-solution whose domain contains $x_0$---whenever it 
terminates.
\item
Assume that {\bf SLR}$_1$ is
applied to a system of pure equations
over a directed set $\D$ where
each right-hand side is monotonic.
If the operator $\Box$ is instantiated with $\upd$,
then for any initial mapping $\rho_0$ and interesting unknown $x_0$, {\bf SLR}$_1$
is guaranteed to terminate and thus always to return
a partial post solution---whenever only finitely many unknowns are encountered.
\end{enumerate}
\end{theorem}

\begin{proof}
We first convince ourselves that, upon termination, each right-hand side can be
considered as being evaluated atomically. For that, we notice that
a call ${\sf solve}\,y$ will never modify the value $\rho[x]$ of an unknown
$x$ with ${\sf key}[x]>{\sf key}[y]$. 
During evaluation of right-hand sides, a recursive call to ${\sf solve}$ may 
only occur for an unknown $y$
that has not been considered before, i.e., is fresh.
Therefore, it will not affect any unknown that has been encountered 
earlier.
From that, we conclude that reevaluating a right-hand side $f_x$ for $\rho$
immediately after a call $f_x\,({\sf eval}\,x)$, will return the same value ---
but by a computation that does not change $\rho$ and thus is atomic.

In order to prove that {\bf SLR}$_1$ is a local generic solver,
it therefore remains to verify that upon termination, $\rho$ is a partial
$\Box$-solution with $x_0\in{\sf dom}$.
Since $x_0$ is initialized before ${\sf solve}\,x_0$ is called,
$x_0$ must be contained in ${\sf dom}$. 
Upon termination, evaluation of no unknown is still
in process and the priority queue is empty. 
All unknowns in ${\sf dom}\setminus{\sf stable}$  
are either fresh and therefore solved right-away,
or non-fresh and then inserted into the priority queue. 
Therefore, we conclude that the equation ${\rho[x]=\rho[x]\;\Box\;f_x\,\rho}$
holds for all $x\in{\sf dom}$.
Furthermore, the invariant for the map ${\sf infl}$ implies that upon termination, 
$x\in{\sf infl}[y]$ whenever $x=y$ or $y\in{\sf dep}_x\,\rho$. In particular,
${\sf infl}$ is defined for $y$ implying that $y\in{\sf dom}$.

In summary, correctness of the algorithm {\bf SLR}$_1$ follows from the stated invariants.
The invariants themselves follow by induction on the number of function calls.
Therefore, statement 1 holds.

For a proof of statement 2, assume that all equations are monotonic and
only finitely many unknowns are
encountered during the call ${\sf solve}\,x_0$.
Let ${\sf dom}$ denote this set of unknowns.
We proceed by induction on key values of unknowns in ${\sf dom}$.
First consider the unknown $x\in{\sf dom}$ with minimal key value. Then
for all mappings $\rho$ and ${\sf infl}$, the call ${\sf solve}\,x$ 
will perform a sequence of updates to $\rho[x]$. In an initial segment of this
sequence, the operator $\upd$ behaves like $\widen$. As soon as the same value
$\rho[x]$ or a smaller value is obtained, the operator $\upd$ behaves like 
the operator $\narrow$. Due to monotonicity, the remaining sequence may only 
consist of narrowing steps. By the properties of widening and narrowing
operators, the sequence therefore must be finite.

Now consider a call ${\sf solve}\,x$ for an unknown $x\in{\sf dom}$ where
by inductive hypothesis, ${\sf solve}\,y$ terminates for all unknowns 
$y$ with smaller keys, and all mappings $\rho$, ${\sf infl}$, sets ${\sf stable}$
and priority queue $Q$ satisfy the invariants of the algorithm.
In particular, this means that every recursive call to a fresh unknown
terminates.

Assume for a contradiction that the assertion were wrong and the call
to ${\sf solve}\,x$ would not terminate. Then this means that 
the unknown $x$ must be destabilized after every evaluation of 
$f_x\,({\sf eval}\,x)$. Upon every successive call to ${\sf solve}\;x$,
all unknowns with keys smaller than ${\sf key}[x]$ are no longer 
contained in $Q$ and therefore are stable.
Again we may deduce that the successive updates for $\rho[x]$ 
are computed by $\widen$ applied to the former value of $\rho[x]$ and
a new value provided by the right-hand side for $x$, until a narrowing
phase starts. Then, however, again due to monotonicity a decreasing
sequence of values for $\rho[x]$ is encountered where each new value
now is combined with the former value by means of $\narrow$.
Due to the properties of $\widen$ and $\narrow$, we conclude that the 
iteration must terminate.
\qed
\end{proof}


\graphicspath{./figs}
\section{Localized $\upd$ in \SLR}
\label{s:localized}
So far we have applied the operator $\upd$ at every right-hand side.
It has been long known for the 2-phase widening and narrowing approach, however, that precision can be gained by applying widening and thus also narrowing only at selected unknowns. These unknowns may be chosen freely, provided they form an \emph{admissible set}, i.e. at least one unknown is selected for each loop in the dependence graph of the equations. 
%
%
%
%
%
%
%
%
When intertwining widening and narrowing by means of structured round-robin or worklist iteration,
restricting $\upd$ to an admissible set of widening points may, however, no longer ensure termination of the 
resulting solvers.

\begin{example}\label{e:wideningpoints}
Consider the same set of equations in the Example~\ref{e:term}. According to our definition, the singleton set
$\{x_2\}$ is admissible. Now assume that the $\upd$ operation is performed for the unknown $x_2$ only. 
With \SRR we obtain the following sequence of updates:
\[
\begin{array}{c|c|c|c|c|c|c|c|c|c|c|c}
i            &\;\;\;\;& 2      & 1      & 2      & \,1\,\, & 3      & 2      & 1   & 2  & 1 &   \\ \hline
x_1          & 0& 0      & \infty & \infty & 1 & 1      & 1      & \infty & \infty & 2 & \dots \\
x_2          & 0& \infty & \infty & 1      & 1 & 1      & \infty & \infty & 2  & 2 & \dots \\ 
x_3          & 0& 0      &      0 & 0      & 0 & 1  & 1 & 1  & 1 & 1 & \dots
\end{array}
\]
Whenever the value for $x_3$ increases, $x_2$ and $x_1$ receive the value $x_3+1$, implying that subsequently, 
$x_3$ further increased. A stable post-solution is never attained. 
A similar behavior can also be observed for \SW on this example.\hfill \qed
\end{example}

\noindent
Example \ref{e:wideningpoints} indicates that we cannot ignore the ordering on the unknowns $x_i$ when selecting the 
points of application for $\upd$. 
Therefore, we refine the notion of admissibility as follows.
%
Assume that we are given a system of equations $x_i = e_i, i=1,\ldots,n$ where sets $\dep(x_i)$ of
variable dependences are explicitly given. 
Then the set $W$ of unknowns is called an \emph{admissible set of $\upd$-points} if, in each cycle in the 
dependence graph of the equations, the unknown with the \emph{highest index} is in $W$.
%
We obtain:

\begin{theorem}
Given a system of equations and an admissible set $W$ of $\upd$-points, both the algorithm $\SRR$ and the algorithm 
$\SW$ is guaranteed to terminate when instantiated with $\Box = \upd$, even when restricting the application of 
$\upd$ to unknowns in $W$ only.
\end{theorem}
\begin{proof}
The proofs are similar to those for the Theorems~\ref{th:SRR} and \ref{th:SW}, respectively.
Here, we only consider the assertion for $\SW$.
For the base case, note that if $x_1$ is the only unknown, 
either the right hand size of $x_1$ is a constant, or it refers to $x_1$ itself, in which case $x_1$ is in the set of 
$\upd$-points. In both cases, \SW terminates. For the inductive case, assume $x_n$ is extracted infinitely many times. 
First assume that $x_n$ is contained in $W$. 
In this case, the proof proceeds as in Theorem~\ref{th:SW}. 
Now assume that $x_n$ is not contained in $W$. Then there is no loop containing $x_n$ which
consists of variables with index at most $n$.
In particular, this means that the set $\{x_1,\ldots,x_{n-1}\}$ can be split into disjoint subsets $X_1,X_2,X_3$.
$X_1$ consists of all unknowns directly or indirectly depending on $x_n$, 
$X_2$ consists of the unknowns onto which $x_n$ directly or indirectly depends, and
$X_3$ contains the remaining unknowns.
As soon as $x_n$ is evaluated for the first time, the evaluation of the unknowns in $X_2$ and $X_3$ have already 
terminated. Therefore following an update of the unknown $x_n$, only unknowns from $X_1$ may be added to the worklist.
Since none of these ever will cause $x_n$ to be added to the worklist again, fixpoint iteration terminates
by the inductive hypothesis.
\hfill\qed
\end{proof}

\begin{example}\label{e:boxpoints}
According to the refined definition, the set $\{x_2\}$ in Example~\ref{e:wideningpoints} is 
no longer admissible, whereas the set $\{x_3\}$ is. When restricting $\upd$ to the latter set, we obtain:
%
\[
\begin{array}{c|c|c|c|c|c|c|c|}
i            &\;\;\;\;   & 2 & 1 & 3      & 2      & 1      & 3        \\ \hline
x_1          & 0& 0 & 1 & 1      & 1      & \infty & \infty \\
x_2          & 0& 1 & 1 & 1      & \infty & \infty & \infty \\ 
x_3          & 0& 0 & 0 & \infty & \infty & \infty & \infty
\end{array}
\]
and the algorithm terminates.\hfill \qed
\end{example}

\noindent
In applications where dependences between unknowns may change, we cannot 
perform any pre-computation on the dependence graph between unknowns. In order to
convieniently deal with these nonetheless,
methods are required which determine admissible sets of $\upd$-points \emph{on the fly}. 
Assume that we are given an assignment $\key$ of unknowns to priorities which
are \emph{linearly} ordered. Such an assignment enables us to 
dynamically identify \emph{back-edges}. Here, a back-edge $x\to y$ consists of 
unknowns $x,y$ where the value of $x$ is queried in the right-hand side of $y$ 
where $\key[x]\geq\key[y]$. Note that this does not correspond to the standard 
definition of back-edge, but we use the same terminology since both may be 
used to identify the head of loops.
%
When a back-edge $x \to y$ is encountered, then $x$ is the unknown with the highest priority 
in some loop and therefore should be included into the set of admissible unknowns, i.e., those
where $\upd$ is going to be applied. 
In all the other case, we may omit the application 
of $\update$.
%
%
The resulting improvement to the solver, as shown in Fig.~\ref{prg:slr+wp}, is called $\SLR_2$.

\begin{figure}
\centering
\begin{tabular}{c}
\begin{lstlisting}[columns=fullflexible,mathescape]
    let rec solve x = 
      `wpx := if x IN wpoint then true else false;`
      if x NOTIN stable then
         stable := stable CUP {x};
         tmp := `if wpx 
           then `RHO[x] UPD f_x (eval x)`
           else f_x (eval x)`
         if tmp NOTEQ RHO[x] then
           RHO[x] := tmp;
           W := `if wpx then infl[x] CUP {x} else` infl[x];
           foreach y IN W do add Q y;
           infl[x] := EMPTYSET;
           stable  := stable MINUS $\;$W;
           while (Q NOTEQ EMPTYSET$\!$) LAND (min_key Q <= key[x]) do
             solve (extract_min Q);
         end
      end 

    and init y = 
      @\emph{as in the original $\SLR_1$}@

    and eval x y = 
      if y NOTIN dom then 
        init y; solve y;
      `if key[x] <= key[y] then wpoint := wpoint CUP {y}`;
      infl[y] := infl[y] CUP {x};
      RHO[y]

    in 
      `wpoint := EMPTYSET`
      @\emph{as in the original $\SLR_1$}@
\end{lstlisting}
\end{tabular}
\caption{\label{prg:slr+wp}The algorithm $\SLR_2$, which is \SLR with plain localized widening. Colored in red are then changes w.r.t. $\SLR_1$.}
\end{figure}

Interestingly for our suite of benchmark programs, the algorithm $\SLR_2$ did not significantly improve the precision of the resulting interval analysis. 
\begin{figure}
\begin{minipage}[b]{0.4\linewidth}
  \centering
  \begin{lstlisting}[xleftmargin=10mm,belowskip=5mm]
    i = 0;
    while (i < 100) {
      j = 0;
      while (j < 10) {
        // Inv: 0 <= i <= 99
        j = j + 1;
      }
      i = i + j;
    }
  \end{lstlisting}
  \caption{Example program with nested loops.}
  \label{f:nested}
\end{minipage}
\hfill
\begin{minipage}[b]{0.55\linewidth}
\begin{center}
\scalebox{0.6}{\input{figs/nested.pspdftex}}
\end{center}
\vspace{-5mm}
\caption{\label{f:nested_cfg}The control-flow graph of the program.}
\end{minipage}
\end{figure}
Consider, e.g., 
the program in Fig.~\ref{f:nested}. The control-flow graph corresponding to this program is shown in Fig.~\ref{f:nested_cfg} where each node $v$ is marked with the priority assigned to $v$ when the function ${\sf solve}$ of $\SLR_1$ is called for the endpoint of the program for an interval analysis. 
%
%
We are looking for nodes that influence nodes with smaller priority. In the example, these are the nodes 
with priorities $-1$ and $-5$, respectively, i.e., exactly the loop heads.
After the first iteration for interval analysis on this program, the interval $[0,0]$ has been established for the program 
variable $i$ at all program points of the inner loop.
Then a second iteration of the outer loop is performed. Even if the operator $\upd$ is only applied at the loop heads,
we obtain the interval $[0,\infty]$ for $i$ at the loop head of the outer 
loop. In the subsequent iteration of the
inner loop, the new interval for variable $i$ at the inner loop head is 
$[0,99]$. Since the operator $\upd$ is meant to be
applied at that program point, the interval $[0,0]\upd[0,99] = [0,\infty]$ is 
recorded for $i$ and subsequently also propagated
to all other program points of the inner loop,
and no subsequent narrowing will take place to recover from the loss of the upper bound for $i$.

This kind of loss of precision is avoided if we allow the set $\wpoint$ of unknowns where to apply $\upd$ 
not only to grow monotonically, but also to \emph{shrink}.
Our second idea therefore is to \emph{remove} an unknown $x$ from $\wpoint$ before the right-hand side of $x$ is evaluated.  
%
%
The resulting algorithm $\SLR_3$ is shown in Fig.~\ref{prg:slr+locw}.
Note that back-edges are detected by the call $\eval\;x\;y$ which therefore may insert $y$ into the set \wpoint, 
while the unknown $x$ is removed from \wpoint inside the call \solve $x$.

%
%
%

\begin{theorem}
When applied to any system of pure equations over a 
directed set $\D$ and interesting unknown $x_0$, the algorithm 
$\SLR_3$ returns a post solution, whenever it terminates. If each right hand 
side is monotonic, then $\SLR_3$ is guaranteed to terminate, whenever only 
finitely many unknowns are encountered.
\end{theorem}
\begin{proof}
The considerations in the original proof for $\SLR_1$ regarding atomicity of 
evaluation of right-hand sides still hold. The same is true for partial 
correctness. The only difference w.r.t.\ $\SLR_1$ is that, upon termination, 
for an unknown $x$ either $\rho[x]= \rho[x] \upd f_x \rho$ or $\rho[x] = f_x 
\rho$. In any case, $\rho$ is a post-solution.

The most interesting part is the proof of termination. So, assume that all right 
hand sides are monotonic and only finitely many unknowns are encountered 
during the call of ${\solve\;x_0}$. Assume the algorithm does not terminate. 
It means there are unknowns $x$ whose values $\rho[x]$ are updated infinitely 
many times. Let $x$ denote one of these unknowns, namely the one with maximum 
priority. From a certain point in the execution of the algorithm, no fresh 
unknown is encountered and  no $\rho[y]$ for an unknown $y$ with key value exceeding 
$\key[x]$ is ever updated.

Assume we have reached this point in the execution of the algorithm. Moreover, 
assume that $x$ is extracted. This means that in the queue there are no
unknowns with key value less than $\key[x]$. Since all unknowns with key values
greater than $\key[x]$ are not subject to update (hence their evaluation does not 
add elements to the queue), for $x$ to be extracted repeatedly, the only 
possibility is that:
\begin{enumerate}
\item in \solve $x$, we should have $tmp \neq \rho[x]$;
\item there is an unknown $y \in \infl[x]$ with $\key[y] \leq \key[x]$, and $y$ is put into the queue.
\end{enumerate}
When $y$ is evaluated (it will happen before \solve $x$ is called again), $x$
will be added to \textsf{wpoint}, hence \textsf{wpx} will always be true when
evaluating \solve $x$. However, by the properties of $\upd$, this means that $x$ cannot 
be updated infinitely many times: contradiction. Therefore the algorithm terminates. \qed
\end{proof}

\begin{figure}
\centering
\begin{tabular}{c}
\begin{lstlisting}[columns=fullflexible,mathescape]
    let rec solve x = 
      wpx := if x IN wpoint then true else false;
      `wpoint := wpoint MINUS $\;${x};`
      if x NOTIN stable then
         stable := stable CUP {x};
         tmp := if wpx 
           then RHO[x] UPD f_x (eval x)
           else f_x (eval x)
         if tmp NOTEQ RHO[x] then
           RHO[x] := tmp;
           W := if wpx then infl[x] CUP {x} else infl[x];
           foreach y IN W do add Q y;
           infl[x] := EMPTYSET;
           stable  := stable MINUS $\;$W;
           while (Q NOTEQ EMPTYSET$\!$) LAND (min_key Q <= key[x]) do
             solve (extract_min Q);
         end
      end 

    and init y = 
      @\emph{as in the $\SLR_1$ and $\SLR_2$}@

    and eval x y = 
      @\emph{as in the $\SLR_2$}@

    in 
      @\emph{as in the original $\SLR_2$}@
\end{lstlisting}
\end{tabular}
\caption{\label{prg:slr+locw}The algorithm $\SLR_3$, which is \SLR with simple localized widening. Colored in red are then changes w.r.t. $\SLR_2$.}
\end{figure}

Let us again consider the program from Fig.~\ref{f:nested}.
The solver $\SLR_3$ iterates through the program points of the inner loop until stabilization before the next iteration on the program points of the outer loop is performed. 
After this iteration, the interval $[0,0]$ has been established for the program variable at all program points of the 
inner loop. 
Since the unknown corresponding to the loop head of the inner loop is now stable, it is no longer contained
in the set {\sf wpoint}. 
Therefore, when during the next iteration of the outer loop the interval 
$[0,99]$ arrives for program variable $i$,
this interval will replace the current interval $[0,0]$ for $i$ (without application of the operator $\upd$).
Accordingly, the subsequent iteration on the inner loop will propagate this interval throughout the inner loop
without change. Therefore no upper bound $\infty$ for $i$ is ever generated within the inner loop. 
This effect is comparable to the concept of localized widening as proposed by \citet{amato2013localizing}.

\section{Restarting in \SLR}
\label{s:restart}
Besides localization of widening and narrowing, \citet{amato2013localizing} present a second idea to improve precision
of fixpoint iteration in presence of infinite increasing chains. Consider the program in Fig.\ \ref{f:hybrid} whose
control-flow graph is given in Fig.~\ref{f:bench}.
In this example, the program variable $i$ takes values from the interval $[0,10]$ whenever the inner loop is entered.
\begin{figure}
  \centering
  \begin{minipage}[b]{0.4\textwidth}
  \begin{lstlisting}[xleftmargin=10mm,belowskip=10mm]
    i = 0;
    while (TRUE) {
      i = i + 1;
      j = 0;
      while (j < 10) {
        // Inv: 1 <= i <= 10
        j = j + 1;
      }
      if (i > 9) i = 0;
    } 
  \end{lstlisting}
  \caption{Example program \textsf{hybrid} from \citep{Halbwachs12}.}
  \label{f:hybrid}
  \end{minipage}
  \hfill
  \begin{minipage}[b]{0.55\textwidth}
    \centering
  \scalebox{0.6}{\input{figs/bench.pspdftex}}
  \caption{\label{f:bench}The control-flow graph for the program from
  Fig.\ \ref{f:hybrid}.}
  \end{minipage}
\end{figure}
The upper bound $10$, though, is missed both by the vanilla version of \SLR as well as of \SLR enhanced with
localized placement of $\upd$. The reason is that the inner loop is iterated with the interval $[1,\infty]$ for $i$ 
until stabilization before, triggered by a narrowing iteration of the outer loop, 
the value $[1,10]$ for $i$ arrives at the entry point of the inner loop.
Since $[1,10]\sqcup[1,\infty]=[1,\infty]$, the finite upper bound of $i$ at the entry point cannot be recovered.

In order to improve on this and similar kinds of precision loss, Amato and Scozzari propose to \emph{restart} 
the iteration for sub-programs. The restart could be triggered, e.g., for the body of a loop as soon as the value 
for the head has decreased.

In the following, we indicate how this strategy may be integrated into the generic solver $\SLR_3$ (see Fig.\ \ref{f:restart}).
The resulting algorithm requires a function {\sf restart}. This function when called with a priority $r$ and an unknown $x$, recursively traverses the ${\sf infl}[x]$ and sets it to the empty set. 
Each found unknown $y$ is added
to the priority queue $Q$ and removed from the set {\sf stable}. 
Moreover, if the priority of $y$ is less than $r$,
then the value $\rho[y]$ is reset to $\rho_0[y]$ and restarting recursively 
proceeds with $r$ and the unknowns from the
set ${\sf infl}[y]$.
The function {\sf restart} then is called within the function {\sf solve} for an unknown $x$ whenever $x$ is currently contained in
{\sf wpoint} and the new value {\sf tmp} for $x$ is less than the current value for $x$. In this case, 
all unknowns in the set ${\sf infl}[x]$ are restarted (w.r.t.~the priority of $x$).
Otherwise, the algorithm behaves like the algorithm $\SLR_3$.

\begin{figure}
  \centering
\begin{minipage}{0.8\textwidth}
  \begin{lstlisting}
    let rec restart r y =
        add Q y;
        stable := stable $\setminus$ {y};
        if key[y] < r then
          RHO[y] := RHO$_0$[y]
          M := infl[y];
          infl[y] := EMPTYSET;
          foreach z IN M do restart r z
    in
    let rec solve x =
      ...
      if tmp NOTEQ RHO[x] then
        `if wpx && tmp SQSUBSETEQ RHO[x] then
          foreach z IN infl[x] CUP {x} do restart key[x] z;
        else` 
          W := if wpx then infl[x] CUP {x} else infl[x];
          foreach y IN W do add Q y;
          stable := stable MINUS $\;$W;
        infl[x] := EMPTYSET;
        RHO[x] := tmp;
        while (Q NOTEQ EMPTYSET$\!$) LAND (min_key Q <= key[x]) do
          solve (extract_min Q);

  \end{lstlisting}
\end{minipage}
\caption{\label{f:restart}Parts of the solver $\SLR_3$ with restarting.}
\end{figure}


Consider again the program from Fig.\ \ref{f:hybrid}. As soon as narrowing the 
head of the outer loop
recovers the interval $[0,9]$ for the program variable $i$, recursively the values for the reachable program points with
lower priorities are reset to $\bot$. This refers to all program points 
in the body of the outer loop and thus also to
the complete inner loop. Reevaluation of all these program points with the value $[0,9]$ for $i$ at the outer loop head 
provides us with the invariant $1\leq i\leq 10$ throughout the inner loop.

The algorithm will return a $\upd$-solution whenever it terminates. 
A guarantee, however, of termination is no longer possible even if right-hand sides are monotonic and
only finitely many unknowns are visited. 
Intuitively, the reason is the following.
Assume that the value for an unknown $x$ has decreased.
Then we might expect that restarting the iteration for lower priority unknowns 
results in a smaller next approximation for $x$. Due to the non-monotonicity 
introduced by widening, this need not necessarily be the case.
Accordingly, we are no longer able to bound the number of switches between increasing and decreasing phases
for $x$.
%
%
There are simple practical remedies for nontermination, though. We may, for example, bound for each unknown
the number of restarts which do not lead to the same value or a decrease.
This behaviour is somewhat different from the restart policy in 
\citep{amato2013localizing} where nontermination  cannot happen, due to the 
fact that the algorithm keeps trace of which (ascending or descending) phase is 
executed in a given program point, and the restart policy cannot transform  a 
descending phase in an ascending phase. 

\section{Side-effecting systems of equations}\label{s:side}
%
%
%
In the following, generic solving, as we have discussed in the 
preceding sections, is extended to right-hand sides $f_x$ that
not only return a value for the left-hand side $x$ of the equation
$x= f_x$, but additionally may produce \emph{side-effects}
to other unknowns. This extension to equation systems, which corresponds to \emph{assert}-statements
of {\sc Prolog} or {\sc Datalog} programs, has been 
advocated in \citep{Apinis12} for an elegant specification of
inter-procedural analysis using partial contexts and 
flow-insensitive unknowns and thus also of multi-threaded programs 
\citep{SVM03}.

\noindent
\begin{example}\label{e:side}
Consider the following program.
\begin{center}
  \begin{minipage}{0.3\textwidth}
  \begin{lstlisting}
    int g = 0;
    void f$\,$(int b) {
      if (b) g =  b + 1;
      else   g = -b - 1;
    }
    int main() {
      f(1);
      f(2);
      return 0;
    }
  \end{lstlisting}
  \end{minipage}
\end{center}
The goal is to determine a tight interval
for the global program variable $g$. 
A flow-insensitive analysis of globals aims at computing
a single interval which should comprise all values
possibly assigned to $g$. 
Besides the initialization with 0, this program 
has two assignments, one inside the call $f(1)$,
the other inside the call $f(2)$. 
A context-sensitive analysis of the control-flow
should therefore collect the three values $0,2,3$ and combine
them into the interval $[0,3]$ for $g$.
This requires to record for which contexts
the function $f$ is called. This task can nicely be
accomplished by means of a local solver. That solver, however,
has to be extended to deal with the contributions to global
unknowns.
\qed

\end{example}
In general, several side effects may occur to the
same unknown $z$. Over an arbitrary domain of values, though,
it remains unclear how the multiple contributions
to $z$ should be combined. Therefore in this section,
we assume that the values of unknowns are taken from
a directed set $\mathbb D$ with a least element, which is denoted by $\bot$. 
Also right-hand sides are again assumed to be
\emph{pure}. For side-effecting constraint systems this means
that evaluating a right-hand side $f_x$ applied to 
functions ${{\sf get}:{\it X}\to{\mathbb D}}$ and 
${{\sf side}:{\it X}\to{\mathbb D}\to{\bf unit}}$, 
consists of a sequence of value lookups for unknowns by means
of calls to the first argument function ${\sf get}$ 
and side effects to unknowns by means of calls to
the second argument function ${\sf side}$ which is
terminated by returning a contribution in $\mathbb D$
for the corresponding left-hand side.

Subsequently, we assume that each right-hand side $f_x$
produces no side effect to $x$ itself and also
to each unknown $z\neq x$ at most one side effect.
Technically, the right-hand side $f_x$ of $x$ with side effects
can be considered as a succinct representation
of a function $\bar f_x$ that takes a mapping $\rho$
and does not return just a single value, but again
another mapping $\rho'$ where 
$\rho'[x]$ equals the return value computed by $f_x$ for ${\sf get}=\rho$,
and for $z\neq x$, $\rho'[z] = d$ if during evaluation of $f_x\:{\sf get}\:{\sf side}$,
${\sf side}$ is called for $z$ and $d$. Otherwise, i.e., if no side effect occurs
to $z$, $\rho'[z]=\bot$.
A post solution of a system ${x=f_x,x\in{\it X}}$,
of equations with side effects then is a mapping
${\rho:{\it X}\to{\mathbb D}}$ such that for every
${x\in{\it X}}$, ${\rho\sqsupseteq \bar f_x\,\rho}$.
A \emph{partial} post solution with domain ${\sf dom}\subseteq{\it X}$
is a mapping $\rho:{\sf dom}\to{\mathbb D}$ such that for every
$x\in{\sf dom}$, evaluation of $f_x$ for $\rho$ accesses only unknowns
in ${\sf dom}$ and also produces side effects only to unknowns in ${\sf dom}$;
moreover, 
$\bar\rho\sqsupseteq \bar f_x\,\bar\rho$ where $\bar\rho$ is the 
total variable assignment obtained from $\rho$ by setting $\bar\rho[y] \gets\bot$
for all $y\not\in{\sf dom}$.

%
%
%

In the following, we present a side-effecting variant ${\bf SLR}_1^+$ 
of the algorithm {\bf SLR}$_1$ from section \ref{s:local}
that for such systems returns a partial  $\Box$-solution---whenever
it terminates. Moreover, the enhanced solver ${\bf SLR}_1^+$
is guaranteed to terminate whenever all right-hand sides $f_x$
are \emph{monotonic}, i.e., the functions $\bar f_x$ all are monotonic. 

\noindent
\begin{example}
Consider again the analysis of example \ref{e:side}.
The contributions to the global program variable $g$
by different contexts may well be combined individually
by widening to the current value of the global.
When it comes to narrowing, though, an individual
combination may no longer be sound. 
Therefore, the extension of the local solver $\SLR_1$
should collect all occurring contributions into a \emph{set},
and use the \emph{joint value} of all these to possibly improve the
value of $g$.
\qed
\end{example}

Conceptually, the algorithm ${\bf SLR}_1^+$ therefore creates for each side effect
to unknown $z$ inside the right-hand side of $x$,
a fresh unknown $\angl{x,z}$ which receives that
single value during evaluation of the right-hand side $f_x$.
Furthermore, the algorithm maintains for every unknown $z$ an auxiliary set
${\sf set}[z]$ which consists of all unknowns $x$ whose right-hand sides
may possibly contribute to the value of $z$ by means of side effects.
Accordingly, the original system of side-effecting equations
is (implicitly) transformed in the following way:
\begin{enumerate}
\item	Inside a right-hand side $f_x$, the
	side effect ${\sf side}\,z\,d$ is implicitly replaced with
	\[
	{\sf side}\,\angl{x,z}\,d
	\]
	while additionally, $x$ is added to the set ${\sf set}[z]$.
\item	The new right-hand side for an unknown $x$ is extended
	with a least upper bound of all $\angl{z,x}$, $z\in{\sf set}[x]$.
\end{enumerate}
The $\upd$-operator is applied whenever 
the return value of the new right-hand side for $x$ is combined with
the previous value of $x$. 
Let us now list the required modifications 
of the algorithm {\bf SLR}$_1$. 

First, the function ${\sf init}\,y$ is extended with an extra initialization of 
the set ${\sf set}[y]$ with $\emptyset$.
The function ${\sf eval}$ remains unchanged.
Additionally, a function ${\sf side}$ is required for realizing
the side-effects during an evaluation of a right-hand side.
As ${\sf eval}$, the function ${\sf side}$ also receives the left-hand side
of the equation under consideration as its first argument. 
We define:
\begin{center}
  \begin{minipage}{0.7\textwidth}
  \begin{lstlisting}
    side x y d = $\:$if LANGx,yRANG NOTIN dom then
                   RHO[LANGx,yRANG] := BOT;
                 if d NOTEQ RHO[LANGx,yRANG] then
                   RHO[LANGx,yRANG] := d;
                   if y IN dom then
                     set[y] := set[y] CUP {x};
                     stable := stable MINUS $\;${y};
                     add Q y
                   else
                     init y;
                     set[y] := {x};
                     solve y
                   end
                 end
  \end{lstlisting}
  \end{minipage}
\end{center}
When called with $x,y,d$, the function ${\sf side}$ 
first initializes the unknown $\angl{x,y}$ if it is not yet contained
in ${\sf dom}$.
If the new value is different from the old value of $\rho$ for $\angl{x,y}$,
$\rho[\angl{x,y}]$ is updated. Subsequently, the set ${\sf set}[y]$ receives the
unknown $x$, and the unknown $y$ is triggered
for reevaluation. 
If $y$ has not yet been encountered, 
$y$ is initialized,
${\sf set}[y]$ is set to $\{x\}$, and ${\sf solve}\,y$ is called.
Otherwise,
$x$ is only added to ${\sf set}[y]$, and $y$ is scheduled for re-evaluation
by destabilizing $y$ first and then inserting $y$ into the priority queue $Q$.

The third modification concerns 
the procedure ${\sf solve}$.
There,
the call of the right-hand side $f_x$ now receives ${\sf side}\,x$
as a second argument and additionally evaluates all unknowns
collected in ${\sf set}[x]$. The corresponding new line reads:
\[
\begin{array}{l}
	{\sf tmp}\gets\rho[x]\upd\,
	(f_x\,({\sf eval}\,x)\;({\sf side}\,x) \sqcup
	\bigsqcup\{\rho[\angl{z,x}]\mid z\!\in\!{\sf set}[x]\});\!\!\\
\end{array}
\]

\noindent
\begin{example}
Consider again interval analysis for the program from example \ref{e:side}.
Concerning the global program variable $g$, 
the initialization $g=0$ is detected first, resulting in the value ${\rho[g]=[0,0]}$.
Then $g$ is scheduled for reevaluation. This occurs immediately, resulting
in no further change.
Then the calls $f(1), f(2)$ are analyzed, the side effects
of $2$ and $3$ are recorded and $g$ is rescheduled for evaluation.
When that happens, the value $\rho[g]$ is increased to 
\[
{[0,0] \upd{} [0,3]} = 
[0,0] \widen{} [0,3] = 
[0,\infty]
\]
if the standard widening for intervals is applied. Since $\rho[g]$ has changed,
$z$ again is scheduled for evaluation resulting in the value
\[
{[0,\infty] \upd{} [0,3]} = 
[0,\infty] \narrow{} [0,3] = 
[0,3]
\]
Further evaluation of $g$ will not change this result any more.
\qed
\end{example}

\noindent
Analogously to theorem \ref{t:local} from the last section, we obtain:

\begin{theorem}\label{t:side}
\begin{enumerate}
\item
When applied to any system of pure equations with side effects and interesting unknown $x_0$, 
the algorithm ${\bf SLR}_1^+$ returns a partial post solution---whenever it terminates.
\item
Assume that ${\bf SLR}_1^+$ is
applied to a system of pure equations
over a directed set ${\mathbb D}$ with bottom, where
each right-hand side is monotonic. Moreover, assume that the $\sqcup$ operator 
is monotonic as well.
Then for any initial mapping $\rho_0$ and interesting unknown $x_0$, $\SLR_1^+$
is guaranteed to terminate and thus always to return
a partial post solution---whenever only finitely many unknowns are encountered 
%
%
and side effects of low priority variables' right-hand sides always refer to higher priority variables.
\end{enumerate}
\end{theorem}

\noindent Note that in the proof of termination we also require the upper 
bound operator 
$\sqcup$ to be monotone. The property trivially holds when $\D$ is a 
join semi-lattice and $\sqcup$ is the least upper bound. However, there are 
some abstract domains which are not join semi-lattices, such as 
zonotopes \citep{DBLP:conf/sas/GoubaultPV12} or parallelotopes 
\citep{DBLP:journals/entcs/AmatoS12}.

The proof of theorem \ref{t:side} is analogous to the proof of theorem \ref{t:local}.
It is worth-while noting, though, that the argument there breaks down if the assumption on the priorities
in side-effects is not met: in that case, any re-evaluation of a high-priority variable $x$ may have 
another effect onto a low-priority variable $y$ --- even if $x$ does not change.
No guarantee therefore can be given that the overall sequence of values for $y$ will eventually become stable.
If on the other hand, the side-effected variable $y$ has priority greater than $x$, 
at re-evaluation time of $y$, the evaluation of $x$ has already terminated where only the final
contributions to $y$ are taken into account.
Since only finitely many such contributions are possible, the algorithm is overall guaranteed to terminate.

The extra condition on the side effects incurred during fixpoint computation is indeed
crucial for enforcing termination --- as can be seen from the following example. 

\noindent
\begin{example}
Consider the following program:
\begin{center}
  \begin{minipage}{0.2\textwidth}
    \begin{lstlisting}
    int g = 0;
    int main() {
      g = g + 1;
      return 0;
    }
    \end{lstlisting}
  \end{minipage}
\end{center}
where the global is meant to be analyzed flow-insensitively.
Consider an interval analysis by means of solver $\SLR_1^+$,
and assume that
the unknown for the global $g$ has lesser priority than the unknown for the endpoint of the
assignment to $g$. The first side effect to $g$ is the interval $[1,1]$ resulting in the
new value $[0,1]$ which is combined
with the old value $[0,0]$ by means of $\upd$ and then again by means of $\upd$.
Since 
\[
([0,0]\upd{}[0,1]) \upd{} [0,1] = [0,\infty]\upd{} [0,1] = [0,1]
\]
the widening is immediately compensated by the consecutive narrowing.
The same phenomenon occurs at every successive update of the value for $g$,
implying that $\SLR_1^+$ will not terminate.

The solver $\SLR_1^+$ behaves differently if the priority of the unknown for $g$ exceeds the
priority of the unknown for the endpoint of the assignment.
In this case after the first application of $\upd$ at $g$, the assignment is 
processed again.
Since the first application of $\upd$ behaves like a widening, this means that the second
side effect to $g$ is with the interval $[1,\infty]$. 
Accordingly, the following recomputation of the new value for $g$ will be
\[
[0,\infty] \upd{} ([0,0] \sqcup [1,\infty]) = [0,\infty] \upd{} [0,\infty] = 
[0,\infty]
\]
and the fixpoint computation terminates.\qed
\end{example}
In practical applications where the side-effected unknowns
correspond to globals, the extra condition on priorities in theorem \ref{t:side}
can be enforced, e.g., by ensuring that the initializers
of globals are always analyzed \emph{before} the call to the procedure ${\sf main}$.

Theorem \ref{t:side} only discusses the extension of the base version of the algorithm $\SLR_1$ to
systems of equations with side effects. A similar extension is also possible to the solvers with localized 
application of $\upd$. In order to ensure termination also in this case, however, 
we additionally must insert every side-effected unknown into the set {\sf wpoint} of unknowns where the operation
$\upd$ is to be applied.
For the side-effecting version of $\SLR_3$, we therefore define:
\begin{center}
  \begin{minipage}{0.55\textwidth}
  \begin{lstlisting}
    side x y d = $\,$wpoint := wpoint $\cup$ {y};
                 if LANGx,yRANG NOTIN dom then
                   RHO[LANGx,yRANG] := BOT;
                 if d NOTEQ RHO[LANGx,yRANG] then
                   RHO[LANGx,yRANG] := d;
                   if y $\in$ dom then
                     set[y] := set[y] CUP {x};
                     stable := stable MINUS $\;${y};
                     add Q y
                   else
                     init y;
                     set[y] := {x};
                     solve y
                   end
                 end
  \end{lstlisting}
  \end{minipage}
\end{center}

\noindent
With this definition, termination of the algorithm $\SLR_3^+$ can be guaranteed under the same assumptions as for
the algorithm $\SLR_1^+$.

\section{Experimental evaluation}\label{s:experimental}

We have implemented the various generic local solvers 
and included into the analyzer {\sc Goblint} for multi-threaded C programs.
{\sc Goblint} uses CIL as C front-end \citep{cil} and is
written in {\sc OCaml}. 
The tests were performed on 2.7GHz Intel Core i7 laptop, with 8GB DDR3 RAM, running OS X 10.9.

In a first series of experiments we tried to clarify the increase of precision possibly attained
by means of the various $\upd$-solvers w.r.t.\ the standard two-phase solving using widening and narrowing
according to \citep{CousotCousot76}.
For these experiments,
we used the benchmark 
suite\footnote{available at \tt www.mrtc.mdh.se/projects/wcet/benchmarks.html}
from the M\"ardalen WCET research group 
\citep{Gustafsson:WCET2010:Benchmarks}
which collects a series of interesting small examples
for WCET analysis, varying in size from about 40 lines to
4000 lines of code.
This benchmark suite we have extended by four tricky programs from 
\citep{amato2013localizing}: 
\begin{enumerate*}
  \item[a)] \texttt{hh.c},
  \item[b)] \texttt{hybrid.c},
  \item[c)] \texttt{nested.c}, and
  \item[d)] \texttt{nested2.c}.
\end{enumerate*}
On top of standard analyses of pointers, we performed an 
interval analysis.
Opposed to the preliminary experiments in \citep{Apinis13}, 
we now use an interval analysis which soundly approximates 32bit integers with
wrap-around semantics. 
For widening, this means that the operator widens the lower and upper bounds first to 
{\sf minint} and {\sf maxint}, respectively,
and, if an underflow or overflow cannot be excluded, also the corresponding upper and lower bounds.
In order to enable two-phase solving, we performed context-insensitive analysis only.

Within this setting, we determined the precision
achieved by the $\upd$-solvers compared to the corresponding solver 
which realizes a distinct widening phase, followed by a distinct narrowing phase.
The results of this comparison is displayed in 
figs.\ \ref{f:precision1}, \ref{f:precision2}, \ref{f:precision3}, and \ref{f:precision4}. 
\begin{figure}
\begin{center}
  \pgfplotsset{width=12.0cm,height=6cm,compat=1.3}
  \begin{tikzpicture} 
  \begin{axis}[
      enlargelimits=0.0,
      ybar interval=0.66,
      yticklabels = {0,0\%,20\%,40\%,60\%,80\%,100\%},
      x tick label style={font=\tiny},
      y tick label style={font=\small},
      xticklabels = {1,,3,,5,,7,,9,,11,,13,,15,,17,,19,,21,,23,,25,,27,,29,,31,,a,b,c,d},
      ybar interval=0.75,
      yticklabel pos=right,
      ymajorgrids=true,
      yminorgrids=true,
      xmajorgrids=false
    ] 
    
  \addplot
    coordinates {
( 1, 83.62) ( 2, 31.83) ( 3, 21.37) ( 4, 75.00) ( 5, 99.43) ( 6, 0.00) ( 7, 81.25) ( 8, 18.87) ( 9, 51.09) (10, 14.63) (11, 93.75) (12, 6.27) (13, 31.03) (14, 26.67) (15, 90.27) (16, 33.33) (17, 43.59) (18, 0.00) (19, 0.00) (20, 86.32) (21, 15.22) (22, 69.28) (23, 28.00) (24, 13.95) (25, 61.54) (26, 64.06) (27, 19.05) (28, 0.00) (29, 39.13) (30, 0.00) (31, 35.71) (32, 88.15) (33, 30.00) (34, 25.00) (35, 22.22) (36, 33.33) (37, 100)
};
    
  \end{axis} 
  \end{tikzpicture}
\caption{\label{f:precision1}The relative improvement of $\SLR_1$ over two-phase solving.}
\end{center}
\end{figure}
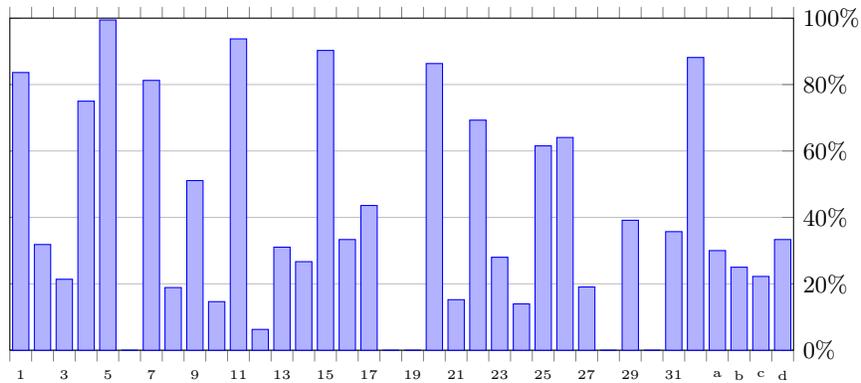
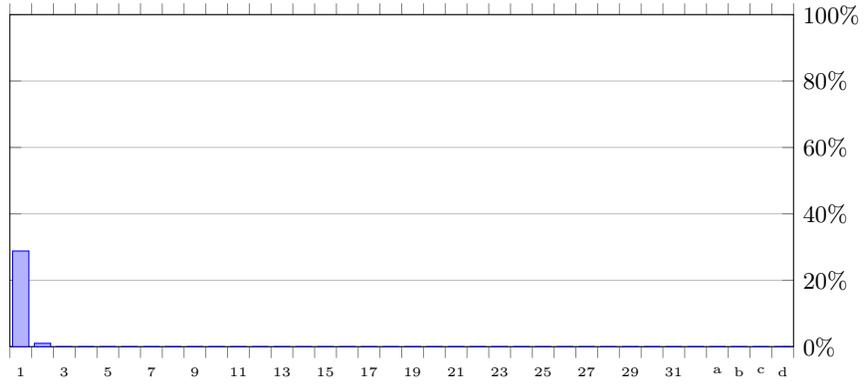
\begin{figure}
\begin{center}
  \pgfplotsset{width=12.0cm,height=6cm,compat=1.3}
  \begin{tikzpicture} 
  \begin{axis}[
      enlargelimits=0.0,
      ybar interval=0.66,
      yticklabels = {0,0\%,20\%,40\%,60\%,80\%,100\%},
      x tick label style={font=\tiny},
      y tick label style={font=\small},
      xticklabels = {1,,3,,5,,7,,9,,11,,13,,15,,17,,19,,21,,23,,25,,27,,29,,31,,a,b,c,d},
      ybar interval=0.75,
      yticklabel pos=right,
      ymajorgrids=true,
      yminorgrids=true,
      xmajorgrids=false
    ] 
    
  \addplot  coordinates {
( 1, 28.81)( 2, 1.01)( 3, 0.00)( 4, 0.00)( 5, 0.00)( 6, 0.00)( 7, 0.00)( 8, 0.00)( 9, 0.00)(10, 0.00)(11, 0.00)(12, 0.00)(13, 0.00)(14, 0.00)(15, 0.00)(16, 0.00)(17, 0.00)(18, 0.00)(19, 0.00)(20, 0.00)(21, 0.00)(22, 0.00)(23, 0.00)(24, 0.00)(25, 0.00)(26, 0.00)(27, 0.00)(28, 0.00)(29, 0.00)(30, 0.00)(31, 0.00)(32, 0.00)(33, 0.00)(34, 0.00)(35, 0.00)(36, 0.00)(37, 100)
};
    
  \end{axis} 
  \end{tikzpicture}
\caption{\label{f:precision2}The relative improvement of $\SLR_2$ over $\SLR_1$.}
\end{center}
\end{figure}
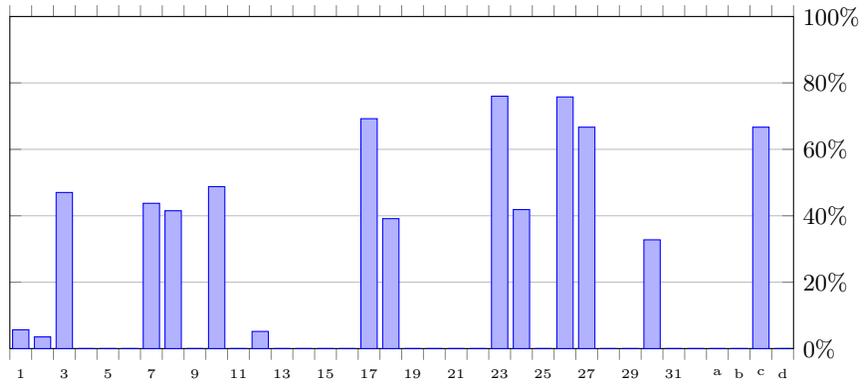
\begin{figure}
\begin{center}
  \pgfplotsset{width=12.0cm,height=6cm,compat=1.3}
  \begin{tikzpicture} 
  \begin{axis}[
      enlargelimits=0.0,
      ybar interval=0.66,
      yticklabels = {0,0\%,20\%,40\%,60\%,80\%,100\%},
      x tick label style={font=\tiny},
      y tick label style={font=\small},
      xticklabels = {1,,3,,5,,7,,9,,11,,13,,15,,17,,19,,21,,23,,25,,27,,29,,31,,a,b,c,d},
      ybar interval=0.75,
      yticklabel pos=right,
      ymajorgrids=true,
      yminorgrids=true,
      xmajorgrids=false
    ] 
    
  \addplot  coordinates {
 (1, 5.65) (2, 3.52) (3, 47.01) (4, 0.00) (5, 0.00) (6, 0.00) (7, 43.75) (8, 41.51) (9, 0.00) (10, 48.78) (11, 0.00) (12, 5.17) (13, 0.00) (14, 0.00) (15, 0.00) (16, 0.00) (17, 69.23) (18, 39.13) (19, 0.00) (20, 0.00) (21, 0.00) (22, 0.00) (23, 76.00) (24, 41.86) (25, 0.00) (26, 75.78) (27, 66.67) (28, 0.00) (29, 0.00) (30, 32.76) (31, 0.00) (32, 0.00) (33, 0.00) (34, 0.00) (35, 66.67) (36, 0.00)(37, 100)
};
    
  \end{axis} 
  \end{tikzpicture}
\caption{\label{f:precision3}The relative improvement of $\SLR_3$ over $\SLR_2$.}
\end{center}
\end{figure}
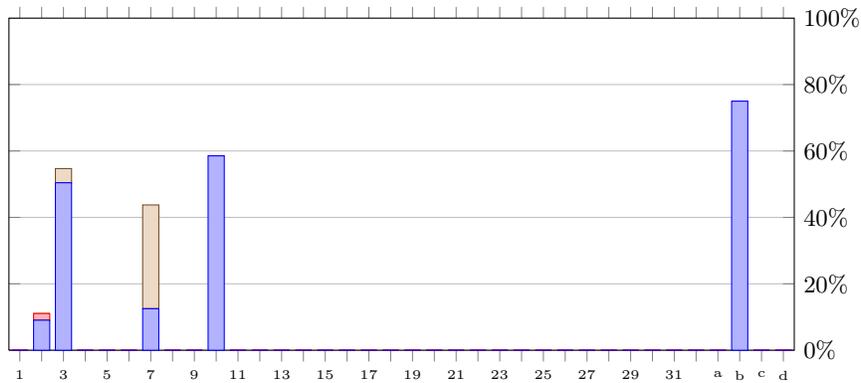
\begin{figure}
\begin{center}
  \pgfplotsset{width=12.0cm,height=6cm,compat=1.3}
  \begin{tikzpicture} 
  \begin{axis}[
      ybar stacked,
      enlargelimits=0.0,
      yticklabels = {0,0\%,20\%,40\%,60\%,80\%,100\%},
      x tick label style={font=\tiny},
      y tick label style={font=\small},
      bar width=2.15mm,
      xtick=data,
      xticklabels = {1,,3,,5,,7,,9,,11,,13,,15,,17,,      
                    19,,21,,23,,25,,27,,29,,31,,a,b,c,d,{}},
      yticklabel pos=right,
      ymajorgrids=true,
      yminorgrids=true,
      xmajorgrids=false,
      x=2.9mm,
      ymax=100,
      ymin=0,
      xmin=0.5,
      xmax=36.5,
      xtick align=outside,
    ] 
    [xshift=3mm]
    \addplot  coordinates { (1,0.00)  (2,9.05)  (3,50.43)  (4,0.00)  (5,0.00)  (6,0.00)  (7,12.50)  (8,0.00)  (9,0.00)  (10,58.54)  (11,0.00)  (12,0.00)  (13,0.00)  (14,0.00)  (15,0.00)  (16,0.00)  (17,0.00)  (18,0.00)  (19,0.00)  (20,0.00)  (21,0.00)  (22,0.00)  (23,0.00)  (24,0.00)  (25,0.00)  (26,0.00)  (27,0.00)  (28,0.00)  (29,0.00)  (30,0.00)  (31,0.00)  (32,0.00)  (33,0.00)  (34,75.00)  (35,0.00)  (36,0.00)  };
   \addplot  coordinates { (1,0) (2,2.01) (3,0) (4,0) (5,0) (6,0) (7,0) (8,0) (9,0) (10,0) (11,0) (12,0) (13,0) (14,0) (15,0) (16,0) (17,0) (18,0) (19,0) (20,0) (21,0) (22,0) (23,0) (24,0) (25,0) (26,0) (27,0) (28,0) (29,0) (30,0) (31,0) (32,0) (33,0) (34,0) (35,0) (36,0) };
    \addplot  coordinates { (1,0.00) (2,0.00) (3,4.27) (4,0.00) (5,0.00) (6,0.00) (7,31.25) (8,0.00) (9,0.00) (10,0.00) (11,0.00) (12,0.00) (13,0.00) (14,0.00) (15,0.00) (16,0.00) (17,0.00) (18,0.00) (19,0.00) (20,0.00) (21,0.00) (22,0.00) (23,0.00) (24,0.00) (25,0.00) (26,0.00) (27,0.00) (28,0.00) (29,0.00) (30,0.00) (31,0.00) (32,0.00) (33,0.00) (34,0.00) (35,0.00) (36,0.00) };    
  \end{axis} 
  \end{tikzpicture}
\caption{\label{f:precision4}Comparison of $\SLR_4$ with $\SLR_3$ indicating the percentage of program points where
the results are incomparable (brown), better (blue) or worse (red).}
\end{center}
\end{figure}
%
Fig.~\ref{f:precision1} reports the percentage of program points where solver $\SLR_1$ returns better results than
two-phase solving.
In the vast majority of cases, $\SLR_1$ returned significantly better 
results---supporting the claim that $\upd$-solving 
may improve the precision.

Fig.~\ref{f:precision2} reports the percentage of program points where an improvement over $\SLR_1$ can be achieved if 
the operator $\upd$ only is applied at widening points, as implemented by solver $\SLR_2$.
Here, our experiments show that, at least for the given simple form of interval analysis, 
an improvement can only be observed for very few example.
The reason might be that, applying narrowing, intertwined with widening can quite often
recover some of the precision lost by the superfluous widenings.

Fig.~\ref{f:precision3} then reports the relative further improvement when additionally widening points can dynamically
be removed during solving. In 15 of 37 cases, we again obtain an improvement, in some cases even for over 70\% 
of program points! This strategy therefore seems highly recommendable to achieve good precision.

Fig.~\ref{f:precision4} finally explores the impact of restarting.
Here, the picture is not so clear. For the second benchmark, restarting resulted even in a loss of precision
for a small fraction of program points, while still for a larger fraction improvements were obtained.
In two further benchmarks, program points with incomparable results where found. 
For benchmark program 3, these make up about 4\% of the program points, while for program 7, the fraction goes even up to
31\%.
In principle such a behavior is not surprising, considering the non-monotonicity of widening.
Still, for two more example programs, drastic improvements are found.
One of these comes from the WCET benchmark suite, while the other has been 
provided in \citep{amato2013localizing},
admittedly, as an example where restarting is beneficial.


\pgfplotsset{width=12cm,height=9cm,compat=1.3}
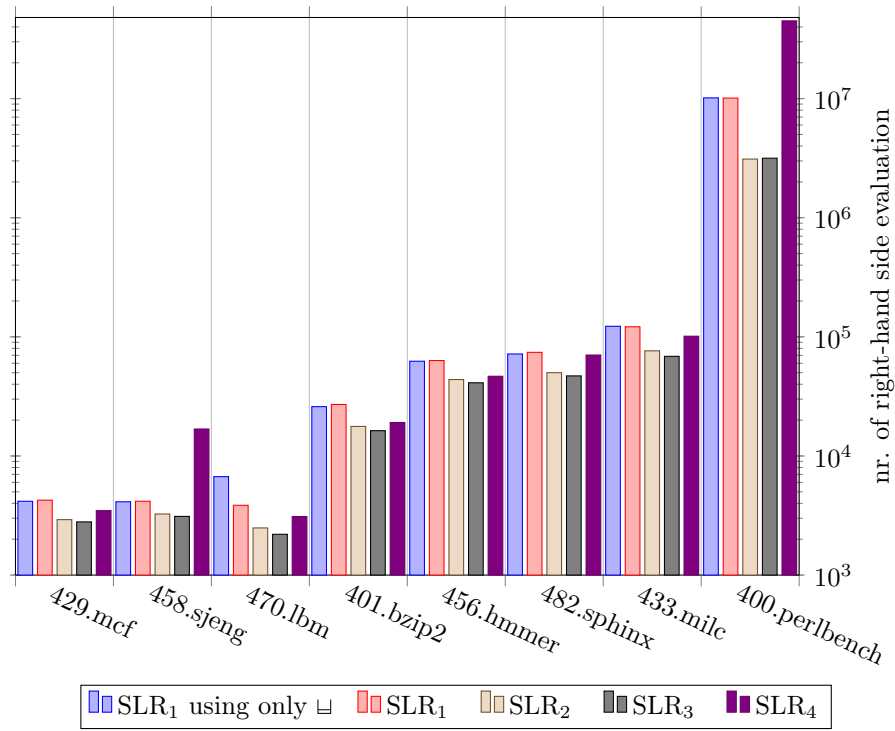
\begin{figure}
  \centering
  \begin{tikzpicture} 
  \begin{semilogyaxis}[
      enlargelimits=0.0,
      legend entries={SLR$_1$ using only $\widen$\ \ \ , SLR$_1$ \ \ \ , SLR$_2$ \ \ \ , SLR$_3$ \ \ \ , SLR$_4$},
      legend to name=legend2,
      legend style={at={(0.5,-0.15)},
      anchor=north,legend columns=-1},
      ybar interval=0.75,
      xticklabels = {429.mcf,458.sjeng,470.lbm,401.bzip2,456.hmmer,482.sphinx,433.milc,400.perlbench,{}},
      x tick label style={rotate=-25,anchor=west,xshift=-0.7em,yshift=-0.4em},
      scaled ticks=false,
      yticklabel pos=right,
      ymax=48000000,
      ytick align=center,
      ylabel={nr. of right-hand side evaluation}
    ]

    \addplot coordinates { (1, 4165) (2, 4125) (3, 6705) (4, 25942) (5, 62446)  (6, 71827)(7, 122670) (8,10150897) (9, 1000)
    };
        
    \addplot coordinates { (1, 4254) (2, 4169) (3, 3854) (4, 27055) (5, 63256) (6, 74146) (7, 121521) (8,10109505)(9, 1000)
    };
      
    \addplot coordinates { (1, 2920) (2, 3256) (3, 2488) (4, 17700) (5, 43824)  (6, 50028)(7, 76350) (8,3109734)(9, 1000)
    };
    
    \addplot coordinates { (1, 2796) (2, 3110) (3, 2200) (4, 16309) (5, 41146)  (6, 47012)(7, 68626) (8,3162658)(9, 1000)
    };
    
    \addplot coordinates {(1, 3475)(2, 16844)(3, 3101)(4, 19096)(5, 46694)(6, 70374)(7, 101295)(8,44978806)(9, 1000)
    };
  \end{semilogyaxis}  
  \end{tikzpicture}
  \ref{legend2}  
  \caption{Context sensitive interval analysis of SpecCpu2006 programs.}
  \label{f:evals2}
\end{figure}

In a second experiment, we explored the relative efficiencies of our
implementation of the generic local $\upd$-solvers. 
For that, we performed interval analysis 
where local variables are analyzed depending on a calling context
which includes all non-interval values of locals,
while the values of globals are analyzed flow-insensitively.
Such kind of analysis cannot be performed by the two-phase approach,
since right-hand sides are not monotonic and
the sets of contexts and thus also the sets of unknowns
encountered during the widening and narrowing phases may vary.

This type of analysis, we tried for all benchmarks from the SpecCpu2006 benchmark suite
which can be handled by the C front-end CIL used in our analyzer. The
set of selected benchmarks consist of seven programs in the range of 1 to 33
kloc, {\tt 400.perlbench} with 175 kloc, and {\tt 445.gobmk} with 412
kloc of C code.
The results for the side-effecting versions of $\SLR_1$ to $\SLR_4$ 
are reported in fig.~\ref{f:evals2} where
the numbers of evaluations of right-hand sides are displayed on a logarithmic scale. 
For a comparison we also included the numbers of evaluations 
if the solver $\SLR_1^+$ uses plain widening instead of $\upd$.

The analysis of the seven smaller programs could be handled in less than 13
seconds. The large program {\tt 400.perlbench} (175 kloc of C code) could be handled by our solvers ---
but with running times between 18 minutes (using $\SLR^+_3$) and 4 hours (using
$\SLR^+_4$), while context-sensitive analysis did not terminate for the largest benchmark {\tt 445.gobmk} ( 412 kloc) 
within 5 hours.

The first observation is that $\SLR_1^+$ is only marginally slowed down,
if widening is enhanced to $\upd$, i.e., narrowing is added.
The second observation is that the efficiency of fixpoint computation
is greatly improved when restricting the application of $\upd$ to widening points.
Improvements of about 30\% could consistently be obtained.
For the large program {\tt 400.perlbench}, the speedup even was by a factor of 3.
%
%
Enhancing solver $\SLR_2^+$ to solver $\SLR_3^+$, on the other hand, which comes with
a significant improvement in precision, additionally results
in another slight reduction of the number of evaluated  right-hand sides.
To us, these numbers came at a surprise, since even in those scenarios where we
could theoretically establish termination of the algorithms, we expected
drastically worse running times of iteration with $\upd$ when compared with iteration with widening alone.

Restarting, finally, adds another dimension of potential inefficiency to
fixpoint iteration. Yet, our numbers for $\SLR_4^+$ on the benchmark suite show
that the practical slowdown over the fastest solver $\SLR_3^+$ is in many cases
still better than solving with $\SLR_1^+$ with widening alone. For the programs
{\tt 458.sjeng} and {\tt 400.perlbench}, however, $\SLR_4^+$ is slower by a
factor of $5$ and $14$, respectively.

%
%
%

In summary, the $\upd$-solver ${\bf SLR}_3^+$ turns out to be a robust algorithm 
with decent run-times. Interestingly, the increase in precision over plain widening 
as well as over $\upd$-solving by means of $\SLR_1^+$, is not penalized by means of a 
slowdown, but is combined with a significant speedup.
The new solver allows to significantly improve precision over the two-phase widening/narrowing
approach and also is successfully applicable in more general analysis scenarios, 
where the two-phase approach was not applicable.

%
%
%
%
%
%
%

\section{Related work}\label{s:related}

Numerous attempts have been made to face the problem of the loss of precision introduced by widening operators.  
Some authors propose to avoid widening and compute a fixpoint of the Kleene 
iteration by using strategy/policy iteration 
\citep{CostanGGMP05-cav,GawlitzaS11-toplas}  or acceleration operators 
\citep{GonnordH06-sas}, but these methods are applicable only to specific 
abstract domains or under syntactical restrictions to the program syntax.
In contrast, our approach is generally applicable, independently from the choice of the 
abstract domain and operators used in the analysis or syntactical restrictions.

Another domain-independent approach is to design enhanced widening operators 
such as delayed widening, widening with threshold  
\citep{blanchet2003static}, widening with landmarks \citep{Simon06} and 
lookahead 
widening \citep{gopan_lookahead_2006}. These may work in some specific settings 
and abstract domains, but still may benefit from an accompanying narrowing iteration. 
These kinds of enhancements are orthogonal to our approach. They may be plugged into the $\upd$-operator, 
and thus be used together with our fixpoint algorithms.

Due to the presence of widening operators, it has been observed that the entire 
analysis fails to be monotonic. Therefore, selecting a different 
starting point of the analysis, other than the bottom of the abstract domain, 
may improve the overall result. In practice, this has been exploited by different techniques, which all have in common the idea to repeat the entire analysis multiple times with some variations, and afterwards
combine the results.
%
%
The proposal of \cite{Halbwachs12} is to iterate the 
analysis starting from a different initial value. After each widening/narrowing 
phase, the result is perturbed in order to get a new value to restart the 
widening/narrowing phase. The intersection of all the obtained results is 
guaranteed to be a post-fixpoint. There are several approaches to choose the 
perturbation, but only the simplest one has been implemented so far. In 
\citep{amato2013localizing}, experimental evidence is provided that localized 
widening with a standard separated narrowing is competitive with respect to 
this approach. Note that $\SLR_2$ generalizes the ideas of 
\cite{amato2013localizing}.
%
%
Gopan and Reps' guided static analysis \citep{gopan_guided_2007} applies a
standard program analysis to a sequence of program restrictions. Each  
restriction is analyzed starting from the result of the previous restrictions, 
until the original program is analyzed. Moreover, 
\cite{HenryMM12-tapas} enhance guided static analysis by 
combining it with path-focusing \citep{MonniauxG11-sas}, in order to avoid 
merging infeasible paths and find precise disjunctive invariants. 
\cite{amato2013localizing} give some evidence, though, that guided static 
analysis 
does not help in those cases 
where localized widening and intertwined widening and narrowing are beneficial.
Monniaux and Le Guen's stratified static analysis by variable dependency  
\citep{stratified} is similar to guided static analysis in that successive 
approximations of the program are considered, where later approximations 
consider more variables than former ones. The result of one approximation is 
used within the successive approximations to improve the results.
%
%

These techniques treat the equation solver as a black box, and try to 
execute different analyses to improve the result. In this sense, they are orthogonal to 
our engineering of fixpoint algorithms and therefore may benefit from our improvements.
In particular, the combination with static guided analysis seems promising.

\section{Conclusion}\label{s:concl}

We have presented a generic combination of widening and narrowing
into a single operator $\upd$ and systematically explored solver algorithms
which, when instantiated with $\upd$ will solve general systems of equations.
Perhaps surprisingly, standard versions of fixpoint algorithms, 
when enhanced with $\upd$,
may fail to terminate even for finite systems of monotonic equations.
Therefore, we presented variants of round-robin iteration, of ordinary worklist
iteration as well as of recursive local solving with and without side effects
where for monotonic equations and finitely many unknowns, termination can be 
guaranteed whenever only finitely many unknowns are encountered, and side-effects 
are to higher-priority unknowns only.
In order to enforce termination, we assigned static priorities to the
unknowns of the system.
In order to construct generic solvers for arbitrary systems of equations,
we heavily relied on self-observation of the solvers. Thus, we assign the priorities in
the ordering in which the unknowns are encountered. We let the fixpoint iterator
itself determine the dependencies between unknowns. Together with the static
priorities, also the places where to apply the operator $\upd$ are dynamically 
determined. 

It has not been clear before-hand, though, how well the resulting algorithms behave for 
real-world program analyses. In order to explore this question, we have provided an implementation within the 
analysis framework Goblint. In our experimental set-up, we considered 
inter-procedural interval analysis
where the monotonicity assumption is not necessarily met. 
Our experiments confirm that fixpoint iteration based on
the combined operator $\upd$ still terminates and may increase precision considerably.
%
This holds true already for the local solver $\SLR_1^+$ which has been 
presented in \citep{Apinis13}.
Beyond that, we demonstrated that the add-on of localizing $\upd$ operators 
increases precision further, while efficiency is improved at the same time.
An equally clear picture could not be identified for the extra optimization of restarting.
While we found improvements in selected cases and generally still an acceptable efficiency,
we also found exceptional cases where a (minor) loss of precision occurs at 
some program points 
or where the performance is degraded considerably.

At the end, we think that the two most important benefits of using the 
$\upd$-operator are:
\begin{itemize}
\item the increase in precision w.r.t.~standard analysis with separate widening and narrowing phases;
\item simpler implementation of solvers w.r.t.~other solutions with separate 
and (especially) interleaved widening and narrowing phases (compare, for 
example, the complexity of the solver based on localized narrowing in 
\citep{amato2013localizing} with the solver \textbf{SRR}).
\end{itemize}

Our experiments were performed for standard interval analysis with the obvious
widening and narrowing operators.
It remains for future work to explore how well our methods work also for other domains
and for more sophisticated widening and narrowing operators.
%
%




\section*{References}
\bibliographystyle{elsarticle-harv}
\bibliography{mybib}

\end{document}